\documentclass[12pt]{article}

\usepackage[T1]{fontenc}
\usepackage{amssymb,amsmath,amsthm,bm,amsfonts,mathrsfs}
\allowdisplaybreaks
\DeclareMathOperator*{\argmin}{argmin}

\usepackage{enumitem}
\usepackage[margin=1in]{geometry}

\usepackage{setspace}
\usepackage{microtype}
\usepackage{booktabs,threeparttable,multirow}

\usepackage[authoryear,longnamesfirst]{natbib}

\usepackage{enumitem}
\usepackage{comment}

\usepackage{caption}
\usepackage{subcaption}
\usepackage{tikz}
\usepackage{pgfplots}
\pgfplotsset{compat=1.17}
\usepgfplotslibrary{fillbetween}

\usepackage[colorlinks = true,linkcolor = blue]{hyperref}
\hypersetup{
     colorlinks = true,
     citecolor = blue,
     urlcolor = blue,
     linkcolor = blue
}

\newtheorem{theorem}{Theorem}
\newtheorem{lemma}{Lemma}
\newtheorem{definition}{Definition}

\title{On Local Overidentification and
Efficiency Gains in Modern Causal Inference and Data Combination\footnote{Authors are listed in alphabetical order. We are grateful to Oliver Linton, Kaspar W\"uthrich and the participants in the June 2025 Econometrics Workshop in honor of Professors Oliver Linton and Yoon-Jae Whang on June 13-14, 2025, Cambridge University for helpful comments.}}
\author{Xiaohong Chen\thanks{Yale University. Email: xiaohong.chen@yale.edu.} \quad Haitian Xie\thanks{Peking University. Email: xht@gsm.pku.edu.cn.}}
\date{First version on arXiv: Oct 19, 2025; Revised version: \today}

\begin{document}

\maketitle

\begin{abstract} 
\singlespacing
This paper studies nonparametric local (over-)identification and the semiparametric efficiency in modern causal frameworks. We develop a unified approach that begins by translating structural models with latent variables into their induced statistical models of observables and then analyzes local overidentification through conditional moment restrictions. We apply this approach to three popular classes of causal models: (1) the general treatment model under unconfoundedness; (2) the negative control model, and (3) the long-term causal inference model under unobserved confounding. The first model yields a locally just-identified statistical model, implying that all regular asymptotically linear estimators of the treatment effect have the same asymptotic variance, which equals the (trivial) semiparametric efficient variance bound. In contrast, the latter two models involve nonparametric endogeneity and are naturally locally overidentified; consequently, some doubly robust orthogonal moment estimators of the average treatment effect are inefficient. Whereas existing work typically imposes strong conditions to restore local just-identification to justify the efficiency of their doubly robust orthogonal moment estimators, we characterize the semiparametric efficient variance bounds, along with efficient estimators, for the (locally) overidentified models (2) and (3). A small real data application, along with a simulation study, illustrates the semiparametric efficiency gains in model (3).


\end{abstract}

\newpage

\section{Introduction}

In the era of generative artificial intelligence (AI) and abundant off-the-shelf machine-learning (ML) tools, it has never been easier to fit flexible models and report “black-box” causal effects. A common workflow estimates a nonparametric object, such as a conditional mean, quantile, or density, using whatever ML packages in a first stage, which is then plugged into an unconditional moment condition to estimate a causal parameter in the second stage. Hidden in this popular workflow is a fundamental econometric point that has been largely overlooked in the modern ML causal literature: local (over-)identification of the model.

Following \cite{chen2018overidentification}, a statistical model is locally just identified at a data distribution when the model’s tangent space spans all valid score directions at that data distribution. Intuitively, this means that near the true data-generating process, any small change might be observed in the data can be matched by the model, so there is no additional information to make estimators more efficient and no locally testable restrictions. 

Modern causal models, however, are typically structural rather than purely statistical: they involve unobserved variables—such as potential outcomes, negative controls, or latent confounders—and identification assumptions formulated as restrictions on the joint distribution of these unobserved variables. Our contribution is to make this distinction explicit and to provide a unified approach to studying local (over-)identification in causal models. We achieve this by translating each structural model into an observationally equivalent statistical model of observables and then analyzing local identification of the observable distribution. This structural-to-statistical translation, combined with the semiparametric efficiency calculation method for general sequential conditional moment restriction models in \citet{ai2012semiparametric}, yields a unified framework for determining efficiency bounds and constructing efficient estimators, thereby avoiding the need for case-by-case analyses across different causal designs.

We use this framework to study three representative modern causal inference models: (1) the general treatment model under unconfoundedness \citep{ai2021unified}, (2) the negative control model \citep{miao2018identifying}, and (3) the long-term causal inference model under unobserved confounding \citep{imbens2025long}. The first model is without nonparametric endogeneity: under unconfoundedness, treatment assignment is as good as random given observed covariates, so the causal effects are determined by comparing observable conditional distributions. In contrast, the negative control model and the long-term causal inference feature nonparametric endogeneity: the treatment effect is identified through solving a nonparametric instrumental variables (NPIV) type inverse problem, in which the nuisance functions appear inside a conditional expectation operator. As shown by \cite{chen2018overidentification}, nonparametric endogeneity typically implies local over-identification, and efficiency considerations become essential in this case. 

The unconfoundedness design, we show, is typically locally just identified. Two consequences follow, regardless of the sophistication of the ML first stage: (i) all regular, asymptotically linear estimators of the treatment effect are first-order equivalent, so the semiparametric efficiency bound is trivial; and (ii) there is no nontrivial specification test for the causal model. In contrast, designs with nonparametric endogeneity are naturally locally overidentified. Rather than imposing extra structure to force local just-identification as in the literature, we directly characterize the efficiency bound in the overidentified case and construct the corresponding efficient estimators. To this end, we formulate these designs as sequential moment condition models and apply \cite{ai2012semiparametric} to obtain efficient influence functions.

Beyond the semiparametric efficiency results in \cite{ai2003efficient,ai2012semiparametric}, a few other recent studies examine the case of overidentification and semiparametric efficiency. \cite{navjeevan2023identification} studies the identification and semiparametric efficiency in a general class of instrumental variables model with effect heterogeneity. \cite{hahn2024overidentification} analyze the overidentifying restrictions that arise in the shift-share (Bartik) instrument designs. \cite{chen2025efficient} derives semiparametric efficient estimators in multi-period difference-in-differences design. 

Overidentification can arise when nuisance functions are specified parametrically or treated as known. Imposing functional form restrictions can lead to local overidentification of the joint distribution of observables, allowing for the existence of estimators that are strictly more efficient than others. A well-known example is the inefficiency of the inverse probability weighting estimator that uses the true propensity score in estimating the average treatment effect \citep{hirano2003efficient}, as well as in more general GMM models \citep{chen2004semiparametric,chen2008semiparametric}. More recently, \citet{CarlsonDell2025} develop a unified framework for robust and efficient estimation with unstructured data that relies on a known propensity score (which they term the “annotation score”). These examples demonstrate that nontrivial efficiency bounds can emerge from functional form restrictions on nuisance functions. At the same time, in practice—especially in the current era of AI and ML—researchers often prefer to estimate nuisance functions nonparametrically using flexible, data-adaptive methods.

Local overidentification can also arise in semiparametric two-stage GMM settings, where the target parameter is defined by potentially overidentified unconditional moment conditions, while the first-stage nonparametric nuisance functions are just identified. In this case, \citet{ackerberg2014asymptotic} show that the semiparametric two-step GMM estimator achieves efficiency.

The remainder of the paper is organized as follows. Section \ref{sec:concepts} formally introduces the concept of local (over-)identification and its connection to semiparametric efficiency gains. Sections \ref{sec:basic}, \ref{sec:negative-control}, and \ref{sec:long-term} examine the general treatment model, negative control, and long-term causal inference models, respectively. Section \ref{sec:numerical} presents an empirical illustration along with a simulation study. Section \ref{sec:conclude} concludes.

\section{Review: Model over-identification and efficiency gains} \label{sec:concepts}

In this section, we first introduce the concepts of structural and statistical models. We then review the key concept of local just-identification and over-identification of a statistical model introduced by \citet{chen2018overidentification}. We also summarize its implications for the existence of efficiency gains in estimation and testing of the regular linear parameter of interest.

\paragraph{Structural and statistical model.} In economic applications, structural variables describe aspects of the data-generating process (DGP) and reflect the researcher’s view of underlying mechanisms. These variables may be hypothetical and exist only in economic theory. In modern causal inference, structural variables encompass potential outcomes under different treatments, potential treatments under different instruments, and possibly other latent variables.

Let $W^*$ denote the vector of structural variables, taking values in $\mathcal{W}^*$. The structural model $\mathbf{P}^*$ is the collection of joint distributions of $W^*$ consistent with the imposed structural assumptions, which capture the researcher’s economic intuition about the setting. The structural variables are not fully observed. We instead observe $W = s(W^*)$, where $s:\mathcal{W}^* \to \mathcal{W}$ is a transformation into the observable space. The statistical model is the set of distributions of $W$ induced by $\mathbf{P}^*$:
\[
   \mathbf{P} = \{ P = P^* \circ s^{-1} : P^* \in \mathbf{P}^* \},
\]
where $P^* \circ s^{-1}$ is the pushforward measure of $P^*$ by the function $s$. 
Although the structural model encodes theoretical restrictions, estimation and inference are always carried out in the induced statistical model.

\paragraph{Model just-identification and over-identification.} 

For any Euclidean set $\mathcal{W}$, denote $\mathbf{M}(\mathcal{W})$ as the set of all probability measures on $\mathcal{W}$, where we always consider the Borel sigma-algebra associated with the Euclidean space.

\begin{definition}
   A statistical model $\mathbf{P}$ is globally just identified if it is fully unrestricted, that is, $\mathbf{P} = \mathbf{M}(\mathcal{W})$. Conversely, $\mathbf{P}$ is globally overidentified if $\mathbf{P}$ is a strict subset of $\mathbf{M}(\mathcal{W})$.
\end{definition}

The concept of global just identification is based on whether there is any restriction imposed on the statistical model $\mathbf{P}$. Although seemingly stringent, the structural assumptions may impose no observable restrictions, so the statistical model remains globally just identified. In Section \ref{sec:basic}, we show that the unconfoundedness model is globally just identified even though assumptions are imposed in the structural model.

The concept of local identification requires the notion of tangent space. Let $L^2_0(P)$ denote the set of mean zero and $P$-square integrable functions, that is,
\begin{align*}
    L^2_0(P) = \Big\{ g:\mathcal{W} \rightarrow \mathbb{R}, \int g dP = 0, \int g^2 dP < \infty \Big\}.
\end{align*}
Take $g \in L^2_0(P)$ and $P \in \mathbf{P}$. A path is a mapping $\theta \mapsto P_{\theta,g}$ from $[0,\bar{\theta})$ to $\mathbf{M}(\mathcal{W})$ such that $P_{0,g} = P$, and 
    \begin{align} \label{eqn:differentiable-quadratic-mean}
        \lim_{\theta \downarrow 0} \int \left( \big( \sqrt{p_{\theta,g}(w)} - \sqrt{p(w)} \big)/\theta - \frac{1}{2}g(w) \sqrt{p(w)}  \right)^2 d\mu_\theta(w) = 0,
    \end{align}
    where $\mu_\theta$ is a $\sigma$-finite positive measure dominating $P_{\theta,g} + P$, and $p_{\theta,g}$ and $p$ denote the respective densities of $P_{\theta,g}$ and $P$.

For the statistical model $\mathbf{P}$, the \emph{tangent space} at $P$ is defined as the closed linear span of the set of feasible scores within $\mathbf{P}$:
\begin{align*}
    \bar{\mathscr{T}}(P) = \textit{cl}\left( \big\{g \in L^2_0(P): (\ref{eqn:differentiable-quadratic-mean}) \text{ holds for some path } \theta \mapsto P_{\theta,g} \in \mathbf{P} \big\} \right),
\end{align*}
where $\textit{cl}$ denotes the closed linear span of a set. By definition, the tangent space $\bar{\mathscr{T}}(P)$ is a subset of $L^2_0(P)$. Whenever this inclusion holds as equality, we say the distribution $P$ is locally just identified by $\mathbf{P}$.

\begin{definition}
   A distribution $P \in \mathbf{P}$ is locally just identified by $\mathbf{P}$ if $\bar{\mathscr{T}}(P) = L^2_0(P)$. Conversely, $P \in \mathbf{P}$ is locally overidentified by $\mathbf{P}$ if $\bar{\mathscr{T}}(P) \subsetneqq L^2_0(P)$.
\end{definition} 

Under global just-identification, $P$ can be approach from any score direction, and hence local just-identification holds given regularity of the model. However, global over-identification does not imply local over-identification, while local over-identification implies global over-identification. See \citet{chen2018overidentification} for details.

A parameter of interest is represented as a functional $\mu:\mathbf{P}\to\mathbb{R}^{d_\mu}$. Following  \cite{AndrewsChenTecchio2025EstimatorPurpose}, we refer to $(\mu,\mathbf{P})$ as an econometric model. A parameter $\mu$ is said to be regular at $P$ if there exists $\psi \in L^2_0(P)$ such that for every path $\theta \mapsto P_{\theta,g}$ with score $g \in T(P)$, 
\begin{align*}
    \frac{d}{d\theta} \mu(P_{\theta,g}) \Big|_{\theta = 0} = \mathbb{E}[ \psi(W)g(W) ].
\end{align*}

\paragraph{Regular asymptotically linear estimator.} An estimator $\hat\mu$ maps the independent and identically distributed (iid) sample $W_1,\ldots,W_n$ into $\mathbb{R}^{d_\mu}$. For a path $\theta \mapsto P_{\theta,g}$, denote $\overset{L_{n,g}}{\rightarrow}$ as convergence in law under $\otimes_{i=1}^n P_{1/\sqrt{n},g}$ and $\overset{L}{\rightarrow}$ as convergence in law under $P^n$. We use $o_P(1)$ to denote a term that converges in probability to zero under $P$. The estimator $\hat{\mu}$ is said to be a regular estimator of $\mu(P)$ if there is a tight random variable $\zeta$ such that 
    \begin{align*}
        \sqrt{n}(\hat{\mu} - \mu(P_{1/\sqrt{n},g})) \overset{L_{n,g}}{\rightarrow} \zeta, 
    \end{align*}
    for any path $\theta \mapsto P_{\theta,g} \in \mathbf{P}$.
The estimator $\hat{\mu}$ is asymptotically linear if
    \begin{align*}
        \sqrt{n}(\hat{\mu} - \mu(P)) = \frac{1}{\sqrt{n}} \sum_{i=1}^n \psi(W_i) + o_P(1),
    \end{align*}
for some $\psi \in L^2_0(P)$. The function $\psi$ is the influence function of the estimator. The efficient influence function is the projection of any influence function onto the tangent space, attaining the minimal covariance matrix among all influence functions.

We emphasize some results from \citet{chen2018overidentification} that we use in this paper: the presence of more efficient estimators of a regular parameter $\mu(P)$ is equivalent to the data distribution $P$ being locally over-identified by the model $\mathbf{P}$. Equivalently, in locally just-identified models, all regular, asymptotically linear estimators are first-order equivalent. In locally just-identified models, any specification test has trivial local power: along any path, the local asymptotic power of a level-$\alpha$ test cannot exceed $\alpha$.

\paragraph{Graphical demonstration.} To build intuition, we illustrate these concepts with a simple two-dimensional graphical demonstration. Consider the statistical variable $W$ whose support contains three points, $\mathcal{W}=\{a,b,c\}$. The set of distributions on $\mathcal{W}$ is equal to
\begin{align*}
    \mathbf{M}(\{a,b,c\}) = \{(p_a,p_b,p_c): p_a,p_b,p_c \in [0,1],p_a+p_b+p_c = 1\},
\end{align*}
which can be equivalently represented as a two-dimensional triangle:\footnote{This triangle of probabilities is often used for demonstration purposes in the literature. It is sometimes referred to as the Marschak-Machina triangle.}
\begin{align*}
    \mathbf{M}(\{a,b,c\}) = \{(p_a,p_b): p_a,p_b \in [0,1],p_a+p_b \leq 1\}.
\end{align*}
For a specific distribution $P = (p_a,p_b,p_c)$, the set of all scores is equal to
\begin{align*}
    L^2_0(P) = \{(g_a,g_b,g_c): p_a g_a + p_b g_b + p_c g_c = 0\},
\end{align*}
which is equivalent to the set of all directions $(g_a,g_b)$ on the two-dimensional plane.

We consider three statistical models in this space:
\begin{align*}
    \mathbf{P}_1 & = \mathbf{M}(\{a,b,c\}), \\
    \mathbf{P}_2 & = \{(p_a,p_b) \in \mathbf{M}(\{a,b,c\}): p_b \geq p_a /2 \}, \\
    \mathbf{P}_3 & = \{(p_a,p_b) \in \mathbf{M}(\{a,b,c\}): p_b = p_a  \}.
\end{align*}
By definition, the model $\mathbf{P}_1$ is globally just identified while the models $\mathbf{P}_2$ and $\mathbf{P}_3$ are globally overidentified. For local identification, we consider a distribution $P = (p_a=0.4,p_b=0.4)$ that belongs to all three statistical models. The distribution $P$ is locally just identified by the models $\mathbf{P}_1$ and $\mathbf{P}_2$ while locally overidentified by the model $\mathbf{P}_3$.

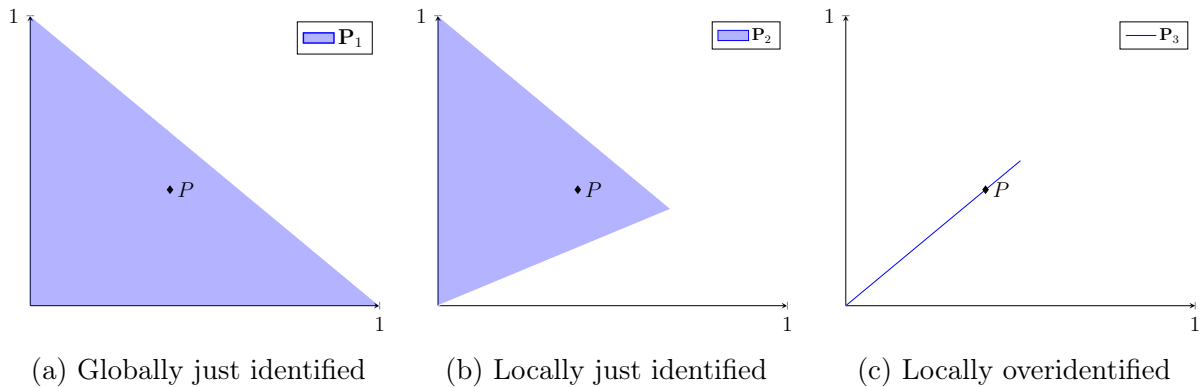
\begin{figure}[!htbp]
    \centering
    \begin{subfigure}[]{0.32\textwidth}
    \resizebox{1\textwidth}{!}{%
    \begin{tikzpicture}
	\begin{axis}[
	legend style={nodes={scale=1, transform shape}},
	axis y line=middle, 
	axis x line=middle,
	y axis line style={opacity=1},
	ytick={0,1},
	xtick={0,1},
	]
				\addplot [
					name path=one,
					forget plot,
					domain=0:1, 
					samples=100, 
					color=white,
					fill opacity=0.3
				]
				{1-x};
				\path[name path=xaxis] (axis cs:0,0) -- (axis cs:1,0);
				\path[name path=yaxis] (axis cs:0,0) -- (axis cs:0,1);
				\addplot [
					thick,
					color=blue,
					fill=blue, 
					fill opacity=0.3
				]
				fill between[
					of=one and xaxis,
					soft clip={domain=0:1},
				];\legend{$\mathbf{P}_1$};
				\addplot[black, mark=diamond*] coordinates{(0.4,0.4)} node[right, color=black] {$P$};
			
				after end axis/.code={
					\path (axis cs:0,0) node [anchor=north west,yshift=-0.075cm,xshift=-0.075cm] {0};
					
				}
				\end{axis}						
    \end{tikzpicture}
    }
    \caption{Globally just identified}
    \end{subfigure}
    \begin{subfigure}[]{0.32\textwidth}
    \resizebox{1\textwidth}{!}{%
    \begin{tikzpicture}
	\begin{axis}[
	legend style={nodes={scale=0.8, transform shape}},
	axis y line=middle, 
	axis x line=middle,
	y axis line style={opacity=1},
	ytick={0,1},
	xtick={0,1},
	]
				\addplot [
					name path=one,
					forget plot,
					domain=0:1, 
					samples=100, 
					color=white,
				]
				{1-x};
				\addplot [
					name path=two,
					forget plot,
					domain=0:1, 
					samples=100, 
					color=white,
				]
				{x/2};
				\path[name path=xaxis] (axis cs:0,0) -- (axis cs:1,0);
				\path[name path=yaxis] (axis cs:0,0) -- (axis cs:0,1);
				\addplot [
					color=blue,
					fill=blue, 
					fill opacity=0.3
				]
				fill between[
					of=one and two,
					soft clip={domain=0:2/3},
				];\legend{$\mathbf{P}_2$};
				\addplot[black, mark=diamond*] coordinates{(0.4,0.4)} node[right, color=black] {$P$};
			
				after end axis/.code={
					\path (axis cs:0,0) node [anchor=north west,yshift=-0.075cm,xshift=-0.075cm] {0};
					
				}
				
				\end{axis}						
    \end{tikzpicture}
    }
    \caption{Locally just identified}
    \end{subfigure}
    \begin{subfigure}[]{0.32\textwidth}
    \resizebox{1\textwidth}{!}{%
    \begin{tikzpicture}
	\begin{axis}[
	legend style={nodes={scale=0.8, transform shape}},
	axis y line=middle, 
	axis x line=middle,
	y axis line style={opacity=1},
	ytick={0,1},
	xtick={0,1},
	]
				\addplot [
					name path=one,
					forget plot,
					domain=0:1, 
					samples=100, 
					color=white,
					fill opacity=0.3
				]
				{x};
				\addplot [
					name path=one,
					domain=0:0.5, 
					samples=100, 
					color=blue,
					fill opacity=0.3
				]
				{x};\legend{$\mathbf{P}_3$};
				\path[name path=xaxis] (axis cs:0,0) -- (axis cs:1,0);
				\path[name path=yaxis] (axis cs:0,0) -- (axis cs:0,1);
				\addplot[black, mark=diamond*] coordinates{(0.4,0.4)} node[right, color=black] {$P$};
			
				after end axis/.code={
					\path (axis cs:0,0) node [anchor=north west,yshift=-0.075cm,xshift=-0.075cm] {0};
					
				}
				
				\end{axis}						
    \end{tikzpicture}
    }
    \caption{Locally overidentified}
    \end{subfigure}
    \caption{Two-dimensional demonstration of model identification.}
    \label{fig:demo}
\end{figure}

Figure~\ref{fig:demo} illustrates the three statistical models. Both $\mathbf{P}_1$ and $\mathbf{P}_2$ are two-dimensional with nonempty interiors, while $\mathbf{P}_3$ is one-dimensional with no interior: the line segment representing $\mathbf{P}_3$ is the path $\theta \mapsto (p_a,p_b)(1+\theta)$. Intuitively, $P$ is locally just identified if and only if it is an interior point of the model. In Figures~\ref{fig:demo}(a)–(b), $P$ lies in the interior, so a parametric submodel can approach $P$ from any direction. In Figure~\ref{fig:demo}(c), $P$ is on the boundary, so submodels can approach only from two directions.

To illustrate the connection to efficiency gains, consider estimating 
$p_a = \mathbb{P}_P(W=a)$ from an iid sample $W_1,\ldots,W_n$. 
In models $\mathbf{P}_1$ or $\mathbf{P}_2$, any regular, asymptotically linear 
estimator has the same asymptotic variance as the sample mean 
$\tfrac{1}{n}\sum_{i=1}^n \mathbf{1}\{W_i=a\}$. 
In contrast, under $\mathbf{P}_3$, the additional restriction $p_a=p_b$ 
allows construction of a strictly more efficient estimator, 
$\tfrac{1}{2n}\sum_{i=1}^n \mathbf{1}\{W_i\in\{a,b\}\}$.

\section{General treatment model under unconfoundedness} \label{sec:basic}

In this section, we introduce the general treatment model with the unconfoundedness assumption, derive the corresponding structural and statistical model, and show that the statistical model is globally just identified. 

In the general treatment model, the treatment variable $T$ has support $\mathcal{T} \subset \mathbb{R}$, which may be discrete, continuous, or mixed. Let $Y(t)$ denote the potential outcome when treatment is assigned at $t \in \mathcal{T}$. For simplicity, we assume that the support of $Y(t)$, $\mathcal{Y} \subset \mathbb{R}$, is a bounded set so that the moments of $Y(t)$ exist.  The conditioning covariate is $X \in \mathcal{X}$. The unconfoundedness assumption states that the treatment choice is conditionally independent of the potential outcomes, that is, $Y(t) \perp T | X$. 
Formally, we define the structural model to be
\begin{align*}
    \mathbf{P}^*_{\operatorname{UC}} = \{ P^* \in \mathbf{M}(\mathcal{Y}^\mathcal{T} \times \mathcal{T} \times \mathcal{X}): Y(t) \perp T | X, \forall t \in \mathcal{T}, \text{ under } P^* \},
\end{align*}
where the subscript UC denotes ``unconfoundedness.''
Any element $P^* \in \mathbf{P}^*_{\operatorname{UC}}$ of the structural model is a plausible joint distribution of $(\{Y(t)\}_{t \in \mathcal{T}},T,X)$ such that the unconfoundedness assumption is satisfied.

The potential outcomes are not observed. Instead, we observe the realized outcome $Y$ based on the treatment choice $T$: $ Y = Y(T).$
The statistical model is the set of distributions of the observable data $(Y,T,X)$ induced by the structural model. Let $s:\mathcal{Y}^\mathcal{T} \times \mathcal{T} \times \mathcal{X} \rightarrow \mathcal{Y} \times \mathcal{T} \times \mathcal{X}$ denote the transformation from structural variables to observed variables, that is, 
\begin{align*}
    s(\{Y(t)\}_{t \in \mathcal{T}},T,X) = (Y(T),T,X).
\end{align*}
The statistical model is formally defined as
\begin{align*}
    \mathbf{P}_{\operatorname{UC}} = \{P^* \circ s^{-1}:P^* \in \mathbf{P}_{\operatorname{UC}}^*\},
\end{align*}
where $P^* \circ s^{-1}$ denotes the pushforward measure of $P^*$ by the function $s$. That is, the statistical model $\mathbf{P}_{\operatorname{UC}}$ consists of all probability measures $P$ of the observable data $(Y,T,X)$ such that $P$ can be induced by some structural distribution $P^* \in \mathbf{P}^*_{\operatorname{UC}}$.

\cite{ai2021unified} consider the following implicitly defined parameter $\mu(P)$:
\begin{align*}
    \mu(P)=\argmin_{\mu\in\mathbb{R}^{d_{\mu}}} \mathbb{E}\!\left[\frac{dP_T(T)}{dP_{T| X}(T| X)}\,L\big(Y-g(T,\mu)\big)\right],
\end{align*}
where $L$ is a convex, differentiable loss function and $g$ is a parametric causal effect function known up to the parameter $\mu$. The terms $dP_{\cdot}$ and $dP_{\cdot|\cdot}$ denote the marginal and conditional density functions, respectively, and the ratio $dP_T/dP_{T| X}$ is defined where $dP_{T|X}>0$ and set to zero otherwise. 

This general treatment model encompasses many prominent models in the causal inference literature as special cases: the average treatment effect (ATE) under binary treatment \citep{hahn1998role,hirano2003efficient}, the quantile treatment effect \citep{firpo2007efficient}, and multivalued treatments \citep{cattaneo2010efficient}. See \citet{ai2021unified} and the references therein. See also \cite{chen2024causal} for efficient estimation of general treatment effects with a diverging number of confounders using modern neural network techniques.

Our analysis focuses on finite-dimensional causal parameters. For inference on fully nonparametric objects such as the conditional average treatment effect (CATE) function, see, for example, \citet{chang2015nonparametric,lee2017doubly}.

The following theorem states our result regarding the general treatment model under unconfoundedness.

\begin{theorem} \label{thm:ATE-basic}
    The statistical model $\mathbf{P}_{\operatorname{UC}}$ is globally just identified, that is, $\mathbf{P}_{\operatorname{UC}} = \mathbf{M}(\mathcal{Y}^{\mathcal{T}} \times \mathcal{T} \times \mathcal{X})$. Hence, the model is locally just identified.
\end{theorem}

Theorem \ref{thm:ATE-basic} implies that for estimating the same causal effect in the general treatment model, different regular, asymptotically linear estimators are first-order equivalent and attain the efficiency bound given by Theorem 1 in \cite{ai2021unified}. 

A subtle point is that ensuring the treatment effect parameter is $\sqrt{n}$-estimable 
requires a regularity condition that the efficient influence function has finite variance.%
\footnote{For example, as shown in \citet{chen2004semiparametric}, 
$\sqrt{n}$-estimability can be achieved under a mild condition on the propensity score, 
which is weaker than assuming it is bounded away from zero.} 
This makes the model globally overidentified. 
However, as shown in the proof of Theorem \ref{thm:ATE-basic}, imposing such a regularity condition does not affect the tangent space, 
and hence the model remains locally just identified.

When the propensity score is known or parametrically specified, the model is locally overidentified. In this case, there exist regular, asymptotically linear estimators that are strictly inefficient; for example, the inverse probability weighting estimator for the average treatment effect using the true propensity score is inefficient \citep{hirano2003efficient}.


\section{Negative control model} \label{sec:negative-control}

The previous model does not involve nonparametric endogeneity and is locally just identified. In this section and the next, we study models with nonparametric endogeneity, represented by NPIV-type conditional moment restrictions. Such models are naturally locally overidentified, as analyzed by \citet{chen2018overidentification}.

Here, we focus on the negative control model studied by \citet{miao2018identifying} and subsequent work. We follow \citet{tchetgen2024introduction} to formulate the potential outcome framework. Let $X$ denote baseline covariates, $D\in\{0,1\}$ the treatment, $V$ the negative-control outcome, and $Z$ the negative-control exposure. For each treatment level $d\in\{0,1\}$ and exposure value $z$ in the support $\mathcal{Z}$ of $Z$, define the potential outcomes $Y(d,z)$ and $V(d,z)$ as the values $Y$ and $V$ would take if, possibly contrary to fact, we set $D=d$ and $Z=z$; likewise define the potential treatment $D(z)$ as the treatment that would be realized if $Z$ were set to $z$. The observed treatment and outcomes are generated according to
\begin{align*}
    D=D(Z),\qquad Y=Y\big(D,Z\big),\qquad V=V\big(D,Z\big).
\end{align*}
Thus, the observed variables are $(Y,D,V,Z,X)$.

The following structural assumptions are maintained for the identification of the average treatment effect.
\begin{enumerate}[label=(\arabic*)]
\item Negative–control outcome: 
$V(d,z)=V$ for all $d,z$.
\item Negative–control exposure: 
$Y(d,z)=Y(d,z')=Y(d)$ for all $d$ and all $z,z'$.
\item Latent ignorability: there exists an unobserved latent variable $U$ such that, for all $d,z$,
$(Y(d),V) \perp (D,Z) | (U,X).$
\item Outcome bridge function: there exists a function $h$ such that
\begin{align} \label{eqn:neg-contr-bridge}
    \mathbb{E}\left[Y | D,X,Z\right]=\mathbb{E}\left[h(V,D,X)| D,X,Z\right].
\end{align}
\item Completeness: Given $D$ and $X$, the distribution of $U$ is complete for $Z$.
\end{enumerate}

The estimand function can be written as
\begin{align*}
    \mu(P) = \mathbb{E}_P[h(V,1,X) - h(V,0,X)],
\end{align*}
which identifies the average treatment effect under the standard overlap condition. 

We denote the statistical model induced by the above identification assumptions as $\mathbf{P}_{\operatorname{NC}}$.
In the following, we show that $\mathbf{P}_{\operatorname{NC}}$ is fully characterized by the NPIV-type moment (\ref{eqn:neg-contr-bridge}). 

\begin{lemma} \label{lm:neg-contr-obs-equiv}
    The family of probability distributions of $(\{(Y(d,z),V(d,z),D(z)):d=0,1,z\in\mathcal{Z}\},U,Z,X)$ satisfying the above identification assumptions is observationally equivalent to the family of probability distributions of $(Y,D,V,Z,X)$ satisfying the moment condition (\ref{eqn:neg-contr-bridge}).
\end{lemma}

Let \( T : L^2(V,D,X) \to L^2(Z,D,X) \) be the conditional expectation operator given by \( (Th) \equiv \mathbb{E}[ h(V,D,X) \mid Z,D,X ] \), and let the adjoint \( T^* : L^2(Z,D,X) \to L^2(V,D,X) \) be \( (T^*g) \equiv \mathbb{E}[ g(Z,D,X) \mid V,D,X ] \).

\begin{theorem} \label{thm:neg-contr-justid}
     A distribution $P$ is locally just identified by the negative control model $\mathbf{P}_{\operatorname{NC}}$ if and only if the closure of the range space of $T$ is the full space:
         \begin{align*}
             \bar{\mathcal{R}} = \textit{cl}(\{ f \in L^2(Z,D,X) : f = Th, \exists h \in L^2(V,D,X) \}) = L^2(Z,D,X),
         \end{align*}
    which is equivalent to the adjoint operator $T^*$ being injective.
\end{theorem}

\citet{cui2024semiparametric} study the semiparametric negative-control model and propose a doubly robust estimator. They show that this estimator achieves semiparametric efficiency when both $T$ and $T^{*}$ are surjective, which in particular implies that $T^{*}$ is injective. Our Theorem \ref{thm:neg-contr-justid} provides a complementary perspective: when $T^{*}$ is injective (equivalently, when $T$ is surjective), the model is locally just identified, so any regular, asymptotically linear estimator attains the efficiency bound.

As noted by \cite{cui2024semiparametric}, requiring both $T$ and $T^*$ to be surjective can be a demanding condition in practice. When $V$ and $Z$ are finitely valued, this requirement implies that their sample spaces must have the same cardinality (as in the case studied by \cite{shi2020multiply}); similarly, when $V$ and $Z$ are continuous, it requires that they have the same dimension. If these requirements are not satisfied, the doubly robust estimator they propose may fail to achieve semiparametric efficiency. In what follows, we build on \cite{ai2012semiparametric} to derive the efficient influence function for the estimand $\mu(P)$ without imposing such restrictions.

The estimation of ATE can be formulated as the following sequential moment restriction model:
\begin{align*}
    \mathbb{E}\left[ h_d(V,X) - \mu_d \right] & = 0, d=0,1, \\
    \mathbb{E}[D(Y - h_1(V,X)) | Z,X] & = 0, \\
    \mathbb{E}[(1-D)(Y - h_0(V,X)) | Z,X] & = 0,
\end{align*}
where $\mu_1$ and $\mu_0$ are, respectively, the mean treated and untreated potential outcomes, and with a slight abuse of notation, we denote the nuisance function as $h=(h_1,h_0)'$, where $h_d(V,X)$ is shorthand for $h(V,d,X),d\in{0,1}$. The parameter ATE is equal to $\mu = \mu_1 - \mu_0$. 

We introduce some notation. Let $e_1 = (1,0)$ and $e_0 = (0,1)$. Define the conditional moment function
$\rho(h;Y,V,D,Z,X) = (D(Y - h_1(V,X)),(1-D)(Y - h_0(V,X)))'$.
Let $\Sigma(Z, X) = \mathbb{E}[\rho \rho' | Z, X]$ denote the conditional variance matrix, and write $\Sigma^{-1}(Z,X)$ for its inverse. For $d = 0,1$, define
$
\Gamma_d(Z, X) = \mathbb{E}[(h_d(V, X) - \mu_d)\rho' \mid Z, X] \Sigma^{-1}(Z,X),
$
which captures the conditional covariance between the moment functions, and is used for orthogonalization. Finally, let $L(Z, X)$ be the diagonal matrix with diagonal entries $-p(Z, X)$ and $-(1 - p(Z, X))$, where $p(Z, X) = \mathbb{E}[D \mid Z, X]$.

\begin{theorem} \label{thm:neg-contr-speb}
    The efficient influence functions for $\mu_1$ and $\mu_0$ in the negative control model are given by
    \begin{align*}
        \mathbb{EIF}(\mu_1) & = h_1(V,X) - \mu_1 - \Gamma_1 \rho - (LTH^{-1}T^*(e_1 - \Gamma_1 L)')'\Sigma^{-1}\rho,
    \end{align*}
    and 
    \begin{align*}
        \mathbb{EIF}(\mu_0) & = h_0(V,X) - \mu_0 - \Gamma_0 \rho - (LTH^{-1}T^*(e_0 - \Gamma_0 L)')'\Sigma^{-1}\rho,
    \end{align*}
    respectively, where $H = L(Z,X)T^* \Sigma^{-1}(Z,X,D) TL(Z,X)$, and $H^{-1}$ is the generalized inverse of $H$.
The efficient influence function for ATE is
    \begin{align*}
        \mathbb{EIF}(\mu) = \mathbb{EIF}(\mu_1) - \mathbb{EIF}(\mu_0).
    \end{align*}
    The semiparametric efficiency bound for ATE is given by the second moment of the efficient influence function.
\end{theorem}

In the proof of Theorem \ref{thm:neg-contr-speb}, we formally demonstrate that the efficient influence function reduces to the one in \cite{cui2024semiparametric} when $T$ and $T^*$ are bijective.

For efficient estimation, one may use either the optimally weighted minimum distance estimator proposed by \cite{ai2012semiparametric} or the influence function–based efficient estimator developed in \cite{chen2023efficient}. As the construction of estimators closely parallels that in the following section, we defer the detailed exposition and refer the reader to that discussion.

\section{Long-term causal inference via data combination} \label{sec:long-term}

Data combination provides a versatile toolkit for addressing measurement error, missing data, and treatment effect problems. See, e.g., \cite{chen2004semiparametric,chen2008semiparametric} for a unified framework. Recent work has focused on identifying and estimating long-term causal effects by combining experimental datasets that report only short-term outcomes with observational datasets that contain both short and long-term outcomes. For example, \cite{chen2023semiparametric} and \cite{athey2025experimental} explicitly show that their models are locally just identified. We now focus on the framework studied by \cite{imbens2025long}.

Following \cite{imbens2025long}, let the treatment indicator be binary, $D\in\{0,1\}$. Each unit has a long-term potential outcome $Y(d)$ and a vector of short-term potential outcomes $S(d)=\bigl(S_{1}(d),S_{2}(d),S_{3}(d)\bigr)$, where the subscript indicates the time of measurement.  The realized outcomes are $Y = Y(D)$ and $S = S(D)$.

Data come from two sources: an experimental sample and an observational sample.  Let $G\in\{E,O\}$ indicate the source.  In the experimental sample $(G=E)$, we observe only the short-term outcomes, whereas in the observational sample $(G=O),$ we observe both long- and short-term outcomes.  Thus, the observed variables are $(Y\,\mathbf 1\{G=O\},\, S,\, D,\, G).$
For clarity, we suppress the covariates from the presentation.

Besides the potential outcomes, there is another variable $U$ that denotes the latent confounders that account
for the association between treatment and potential outcomes in the observational data.

The following structural assumptions are imposed by \citet{imbens2025long}:
\begin{enumerate}[label=(\arabic*)]
    \item Observational data: the latent $U$ accounts for all confounding in the observational data: $(Y(d),S(d))\perp D | U, G=O$.
    \item Experiment: the treatment is independent of the potential outcomes and latent $U$ in the experimental data: $(Y(d),S(d),U) \perp D | G=E$.
    \item External validity: the short-term potential outcomes and latent $U$ are independent with the data source $G$: $(S(d),U) \perp G$.
    \item Sequential outcomes: the long-term $Y(d)$ and last period $S_3(d)$ are independent with the first period $S_1(d)$ conditional on the second period $S_2(d)$ and latent $U$ in the observational data: $(Y(d),S_3(d)) \perp S_1(d) | S_2(d),U, G=O$.
    \item Completeness: given $S_2,D,$ and $G=O$, the distribution of latent $U$ is complete for $S_1$.
    \item Outcome bridge function: there exists a function $h$ such that 
    \begin{align} \label{eqn:outcome-bridge-unobs}
        \mathbb{E}[Y|S_2,D,U,G=O] = \mathbb{E}[h(S_3,S_2,D)|S_2,D,U,G=O].
    \end{align}
\end{enumerate}

Under the maintained assumptions, \citet{imbens2025long} show that the bridge function also satisfies the following observable conditional moment:
\begin{align} \label{eqn:outcome-bridge-obs}
    \mathbb{E}[Y|S_2,S_1,D,G=O] = \mathbb{E}[h(S_3,S_2,D)|S_2,S_1,D,G=O].
\end{align}
Denote the above conditional expectation operator by $K$. That is, for any $h \in L^2(S_3,S_2,D)$, $(Kh)(S_2,S_1,D) = \mathbb{E}[\mathbf{1}\{G=O\}h(S_3,S_2,D)|S_2,S_1,D]$. Denote its adjoint operator as $K^*$, i.e., $(K^*g)(S_3,S_2,D) = \mathbb{E}[\mathbf{1}\{G=O\}g(S_2,S_1,D)|S_3,S_2,D]$.
In fact, any $h$ satisfying (\ref{eqn:outcome-bridge-obs}) satisfies (\ref{eqn:outcome-bridge-unobs}), which leads to the identification of the Average Long-Term Treatment Effect (ALTTE) for the observational population:\footnote{In addition to the outcome-bridge approach, \citet{imbens2025long} introduce an alternative identification strategy, the selection-bridge, for identifying the long-term treatment effect.  When the two bridge functions are used jointly, the resulting identification is ``doubly robust.''  Since our paper focuses on semiparametric efficiency rather than the identification of causal parameters, we omit the analysis of the selection bridge function to maintain focus.}
\begin{align*}
    \mu(P) = \mathbb{E}_P[h(S_3,S_2,D)|D=1,G=E] - \mathbb{E}_P[h(S_3,S_2,D)|D=0,G=E].
\end{align*}

We denote the statistical model induced by the above identification assumptions as $\mathbf{P}_{\operatorname{LT}}$.
In the following, we show that $\mathbf{P}_{\operatorname{LT}}$ is fully characterized by the nonparametric-IV-type moment (\ref{eqn:outcome-bridge-obs}). This result can be used to characterize when the model is locally just identified.

\begin{lemma} \label{lm:long-term-obs-equiv}
    The family of probability distributions of $(\{Y(d),S(d):d=0,1\},U,D,G)$ satisfying the above identification assumptions is observationally equivalent to the family of probability distributions of $(Y\mathbf{1}\{G=O\},S,D,G)$ satisfying the moment condition (\ref{eqn:outcome-bridge-obs}).
\end{lemma}

\begin{theorem} \label{thm:long-term-justid}
     A distribution $P$ is locally just identified by the long-term causal inference model $\mathbf{P}_{\operatorname{LT}}$ if and only if the closure of the range space of $K$ is the full space:
         \begin{align*}
             \bar{\mathcal{R}} = \textit{cl}(\{ f \in L^2(S_2,S_1,D) : f = Kh, \exists h \in L^2(S_3,S_2,D) \}) = L^2(S_2,S_1,D),
         \end{align*}
    which is equivalent to the adjoint operator $K^*$ being injective.
\end{theorem}

Theorem 7 in \citet{imbens2025long} characterizes the semiparametric efficiency bound and proposes an efficient estimator for ALTTE under the condition that the operator $K$ is bijective. In this case, the adjoint operator $K^{*}$ is injective, so the model is locally just identified according to our Theorem \ref{thm:long-term-justid}, and any regular, asymptotically linear estimator attains the efficiency bound. In more practical settings, however, $K^{*}$ may fail to be injective—e.g., when $S_{3}$ is not sufficiently rich to capture the full variation in $S_{1}$—leading the model to be locally overidentified. Our results highlight that in such cases, additional efficiency gains are possible beyond the estimator of \citet{imbens2025long}.

To obtain the semiparametric efficiency bound in the general setting in which $K$ is not necessarily bijective, we formulate the estimation of ALTTE as the following sequential moment restriction model:
\begin{align*}
    \mathbb{E}\left[ \mathbf{1}\{G=E\} D \left( h(S_3,S_2,D) - \mu_1 \right) \right] & = 0, \\
    \mathbb{E}\left[ \mathbf{1}\{G=E\} (1-D) \left( h(S_3,S_2,D) - \mu_0 \right) \right] & = 0, \\
    \mathbb{E}[\mathbf{1}\{G=O\}(Y - h(S_3,S_2,D)) | S_2,S_1,D] & = 0,
\end{align*}
where $\mu_1$ and $\mu_0$ are respectively the mean treated/untreated potential outcome. The parameter ALTTE is equal to $\mu = \mu_1 - \mu_0$. 

\begin{theorem} \label{thm:long-term-speb}
    The efficient influence functions for $\mu_1$ and $\mu_0$ in the long term causal inference model are given by
    \begin{align*}
        \mathbb{EIF}(\mu_1) & = \frac{\mathbf{1}\{G=E\} D\left( h(S_3,S_2,D) - \mu_1 \right)}{\mathbb{P}(G=E,D=1)} \\
        & + \frac{ \mathbf{1}\{G=O\}[KM^{-1} \pi(S_3,S_2,D)] D \Sigma_2^{-1}(Y - h(S_3,S_2,D))}{\mathbb{P}(G=E,D=1)}
    \end{align*}
    and 
    \begin{align*}
        \mathbb{EIF}(\mu_0) & = \frac{\mathbf{1}\{G=E\} (1-D)\left( h(S_3,S_2,D) - \mu_0 \right)}{\mathbb{P}(G=E,D=0)} \\
        & + \frac{ \mathbf{1}\{G=O\}[KM^{-1} \pi(S_3,S_2,D)] (1-D)\Sigma_2^{-1}(Y - h(S_3,S_2,D))}{\mathbb{P}(G=E,D=0)},
    \end{align*}
    respectively, where $M = (K^* \Sigma_2^{-1}(S_2,S_1,D) K)$, $M^{-1}$ is the generalized inverse of $M$, $\Sigma_2$ is the conditional variance of the conditional moment:
\begin{align*}
    \Sigma_2(S_2,S_1,D) = \mathbb{E}[\mathbf{1}\{G=O\}(Y - h(S_3,S_2,D))^2 | S_2,S_1,D],
\end{align*} 
and $\pi(S_3,S_2,D) = \mathbb{P}(G=E|S_3,S_2,D)$.
The efficient influence function for ALTTE is
    \begin{align*}
        \mathbb{EIF}(\mu) = \mathbb{EIF}(\mu_1) - \mathbb{EIF}(\mu_0).
    \end{align*}
    The semiparametric efficiency bound for ALTTE is given by the second moment of the efficient influence function:
    \begin{align} \label{eqn:speb-long-term}
        & \frac{1}{\mathbb{P}(G=E)^2} \mathbb{E}\left[ \mathbf{1}\{G=E\} \left( \frac{D - p_E}{p_E(1-p_E)} (h(S_3,S_2,D) - \mu(D))\right)^2 \right] \nonumber \\
        + & \frac{1}{\mathbb{P}(G=O)^2} \left\lVert M^{-1/2} \left( \frac{D - p_O}{p_O(1-p_O)} \frac{\mathbb{P}(G=O|D)}{\mathbb{P}(G=E|D)} \pi(S_3,S_2,D) \right) \right\lVert^2,
    \end{align}
    where $p_E = \mathbb{P}(D=1|G=E)$, $p_O = \mathbb{P}(D=1|G=O)$, $\mu(D) = D\mu_1 + (1-D)\mu_0$, $\lVert \cdot \rVert$ is the $L^2$-norm in the Hilbert space.
\end{theorem}

    The first component of the efficiency bound in (\ref{eqn:speb-long-term}) coincides with the corresponding term in Theorem 7 of \cite{imbens2025long}, whereas the second component generally differs. However, when $K$ is bijective, this second term reduces to the corresponding expression in Theorem 7 of \cite{imbens2025long}.\footnote{A formal proof of this statement is given in the proof of Theorem \ref{thm:long-term-speb}.}

Below, we introduce two efficient estimators for the ALTTE. First, we consider the efficient influence function-based estimator proposed by \citet{chen2023efficient}, which automatically satisfies the Neyman orthogonality condition with respect to the nuisance parameters. We focus on the efficient estimator for $\mu_1$; the estimator for $\mu_0$ is symmetric. Taking the difference between the two estimators yields an efficient estimator for the ALTTE. In the proof of Theorem \ref{thm:long-term-speb}, we have shown that the efficient influence function can be written in the following form:
\begin{align*}
    \mathbb{EIF}(\mu_1) & = \frac{\mathbf{1}\{G=E\} D\left( h(S_3,S_2,D) - \mu_1 \right)}{\mathbb{P}(G=E,D=1)} \\
        & - \frac{ \mathbf{1}\{G=O\}\mathbb{E}[\mathbf{1}\{G=O\}v^*|S_2,S_1,D] \Sigma_2^{-1}(Y - h(S_3,S_2,D))}{\mathbb{P}(G=E,D=1)},
\end{align*}
where $v^*$ is the Riesz representer defined as
\begin{align} \label{eqn:vstar}
    v^* = \frac{\mathbb{P}(G=E,D=1)^2\text{Var}(h(S_3,S_2,D) | G=E,D=1)}{\mathbb{P}(G=E,D=1) + \mathbb{E}[\pi(S_3,S_2,D)D r^*(S_3,S_2,D)]} r^*(S_3,S_2,D),
\end{align}
with $r^*$ being the minimizer of the following objective function
\begin{align} \label{eqn:r0}
    & \frac{\left(\mathbb{P}(G=E,D=1) + \mathbb{E} \left[ \mathbf{1}\{G=E\} D r(S_3,S_2,D) \right] \right)^2}{\mathbb{P}(G=E,D=1)\text{Var}(h(S_3,S_2,D) | G=E,D=1))} \nonumber \\
    + & \mathbb{E} \left[ (\mathbb{E}[\mathbf{1}\{G=O\}r(S_3,S_2,D)|S_2,S_1,D])^2 \Sigma_2^{-1}(S_2,S_1,D) \right].
\end{align}
In implementation, given an iid sample $(Y_i \mathbf{1}\{G_i=O\}, S_i, D_i, G_i)$, one can first obtain a preliminary consistent estimator $\hat{h}$ for $h$, for example, using the conditional moment condition with identity weighting. Then, we can construct the following efficient influence function-based estimator:
\begin{align*}
    \hat{\mu}_1 & = \frac{\sum_{i=1}^n \mathbf{1}\{G_i=E\} D_i \hat{h}(S_{3i},S_{2i},D_i) }{\sum_{i=1}^n \mathbf{1}\{G_i=E\}D_i} \\
    & + \frac{\sum_{i=1}^n \mathbf{1}\{G_i=O\}\mathbb{E}[\mathbf{1}\{G_i=O\}\hat{v}^*|S_{2i},S_{1i},D_i] \hat{\Sigma}_2^{-1}(S_{2i},S_{1i},D_i)(Y_i - \hat{h}(S_{3i},S_{2i},D_i))}{\sum_{i=1}^n \mathbf{1}\{G_i=E\}D_i},
\end{align*}
where $\hat{v}^*$ is estimated as the sample analogue of (\ref{eqn:vstar}) and (\ref{eqn:r0}) over the sieve space.

Second, we construct an estimator based on the optimally weighted minimum distance method. The analysis of the estimator
as a minimum distance problem is a specialization of \cite{ai2003efficient,ai2007estimation,ai2012semiparametric,chenpouzo2015sieve}. We first obtain a consistent estimator $\hat{\Sigma}_2$ using a preliminary estimate of $h$, and then estimate $\hat{h}$ by minimizing
\begin{align*}
\frac{1}{n} \sum_{i=1}^n \hat{\Sigma}_2^{-1}(S_{2i},S_{1i},D_i)\left[\hat{\mathbb{E}}[(Y_i - \hat{h}(S_{3i},S_{2i},D_i)) \mid S_{2i},S_{1i},D_i]\right]^2
\end{align*}
over the sieve space, where $\hat{\mathbb{E}}[\cdot \mid S_2, S_1, D]$ denotes a consistent nonparametric estimator of the corresponding conditional expectation. The resulting estimator for $\mu_1$ is
\begin{align*}
\hat{\mu}_1 = \frac{\sum_{i=1}^n \mathbf{1}\{G_i=E\} D_i \hat{h}(S_{3i},S_{2i},D_i)}{\sum_{i=1}^n \mathbf{1}\{G_i=E\} D_i}.
\end{align*}

\section{Numerical studies} \label{sec:numerical}

\subsection{Empirical application: long-term effect of job training} 

In this section, we apply the efficient estimator developed in Section~\ref{sec:long-term} to study the long-term effect of job training on earnings. We use the Job Corps (JC) dataset as the experimental sample and the Survey of Income and Program Participation (SIPP) as the observational sample.

The JC dataset is obtained from the R package \texttt{causalweight} \citep{causalweight} and originates from the Job Corps National Study, a large-scale randomized evaluation conducted by the US Department of Labor \citep{schochet2008job}. It reports random assignment to the training program, which we take as the treatment $D$, as well as post-treatment log monthly earnings measured one, two, and three years after assignment, which form the short-term outcomes $S_1, S_2, S_3$. The covariates include years of schooling, age, race, and sex at birth.

The SIPP data are publicly available from the US Census Bureau.\footnote{See the Survey of Income and Program Participation, US Census Bureau, \url{https://www.census.gov/programs-surveys/sipp.html}.}
 We focus on the 1996 to 2000 panel to ensure close temporal alignment with the JC study. This panel consists of twelve waves (each wave corresponds to a four-month period), with training-related information reported in the topical module of Wave~2. We define the treatment $D$ as participation in any training intended to help train for a new job or improve skills in the current job during the past year. Short-term outcomes $(S_1,S_2,S_3)$ are defined as earnings measured in Waves~2, 5, and~8. The long-term outcome $Y$ is earnings measured in the final wave of the panel. The set of covariates is chosen to match those used in the JC dataset. In the end, the observational sample contains 28{,}232 untreated and 9{,}629 treated observations, while the experimental sample contains 3{,}663 untreated and 5{,}577 treated observations.

We implement both the estimator of \cite{imbens2025long} and the efficient estimator proposed in Section~\ref{sec:long-term}. The sieve spaces for the nuisance function $h$ are constructed using fourth-order spline bases. The probability $\pi$ is estimated using a sieve logit with spline basis, and the conditional covariance matrix is estimated via $K$-nearest neighbors (KNN). The results are reported in Table \ref{tab:empirical}.

\begin{table}[htbp]
\centering
\caption{Estimates of the long-term effect}
\label{tab:empirical}
\begin{tabular}{lccc}
\toprule
 & IKMW & Efficient & Hausman test \\
\midrule
Estimate & 0.198 & 0.181 &  \\
(SE)     & (0.036) & (0.028) & [0.450] \\
\bottomrule
\end{tabular}
\begin{tablenotes}
\footnotesize
\item Notes: The table reports estimates of the long-term effect. The first two columns correspond to the estimator of \cite{imbens2025long} and our efficient estimator, respectively. Standard errors are reported in parentheses. The Hausman test column reports $p$ values in square brackets.
\end{tablenotes}
\end{table}

Both estimators indicate that job training increases long term earnings by nearly 20\%, a magnitude that is broadly consistent with existing evidence in the literature \citep{schochet2008job}.
 The efficient estimator achieves a reduction in the standard error exceeding 20\%, reflecting a substantial gain in statistical precision. The Hausman test yields no evidence against the model specification of this empirical study.\footnote{In cases of model misspecification, one can utilize the theory developed in \citet{ai2007estimation}.}

\subsection{Simulations}

We run a simulation study with data-generating processes (DGP) that follows the simulation study in \cite{imbens2025long}. We first generate four covariates $\tilde X$ from a multivariate normal distribution with mean zero and correlation coefficients calibrated to match those estimated in the empirical application. We also generate latent unobservables $U$ from a multivariate normal distribution with mean zero, variance $0.5$, and zero covariance across components.

Potential intermediate and final outcomes are generated according to the recursive system
\begin{align*}
Y(a) &= \tau_y a + \alpha_y^{\top}\tilde S_3(a) + \beta_y^{\top}\tilde X + \gamma_y^{\top} U + \epsilon_y, \\
\tilde S_j(a) &= \tau_j a + \alpha_j \tilde S_{j-1}(a) + \beta_j^{\top}\tilde X + \gamma_j^{\top} U + \epsilon_j, \qquad j\in\{3,2\}, \\
\tilde S_1(a) &= \tau_1 a + \beta_1^{\top}\tilde X + \gamma_1^{\top} U + \epsilon_1,
\end{align*}
where $\tau_y$ and $\tau_j$ are scalars and $(\alpha_y,\beta_y,\gamma_y)$ and $(\alpha_j,\beta_j,\gamma_j)$ are scalars, vectors, or matrices of conformable sizes. Following \cite{imbens2025long}, we draw the entries of $\tau_y$, $(\tau_j,\alpha_y,\beta_y,\gamma_y)$, and $(\alpha_j,\beta_j,\gamma_j)$ from the uniform distribution on $[0,1]$ and rescale them so that the $\ell_2$ norms of the vectors $(\tau_j,\alpha_y,\beta_y,\gamma_y)$ and the columns of $(\alpha_j,\beta_j,\gamma_j)$ are equal to $0.5$. The disturbances $\epsilon_j$ are independent mean zero Gaussian variables with variance $0.5$. The error term $\epsilon_y$ has standard deviation $\exp(|\tilde X_1|)$ to introduce heteroskedasticity.

In the experimental sample, treatment is randomly assigned with probability one half. In the observational sample, treatment follows a logistic specification $\mathbb P(D=1 \mid \tilde X,U,G=O)
= \left(1+\exp(\kappa_1^{\top}\tilde X+\kappa_2^{\top}U)\right)^{-1},$
where $(\kappa_1,\kappa_2)$ are generated using the same sampling and rescaling scheme. This induces persistent confounding through the dependence of treatment on both observed covariates and latent factors.

To introduce nonlinear bridge functions, we apply the elementwise transformation
\begin{align} \label{eqn:q}
    \mathrm{sign}(t)\lvert t\rvert^{q}, \qquad q\in\{1,2\},
\end{align}
to each component of $\tilde X,\tilde S_1,\tilde S_2,\tilde S_3$, yielding the observed variables $(X,S_1,S_2,S_3)$. This transformation is invertible and preserves the one to one correspondence between latent and observed variables. When $q=1$, the bridge functions are linear in $(X,S_1,S_2,S_3)$, whereas for $q=2$ they become nonlinear, allowing us to evaluate performance under both linear and nonlinear specifications. We consider $\lambda\in\{0.33, 0.50, 0.67\}$ as the share of experimental data in the sample.

The nuisance functions are estimated using spline bases with four knots and cubic degree. Ridge regularization is applied to the outcome regression to incorporate modern machine learning techniques, and the conditional variances are estimated via KNN.\footnote{Other choices of the tuning parameters deliver similar results.} Table \ref{tab:sim} reports the simulation results based on 500 Monte Carlo replications. As shown in the table, the efficient estimator achieves a reduction in variance across all designs. The dispersion of both estimators declines with the sample size at the root-$n$ rate, consistent with the asymptotic theory.

\begin{table}[htbp!]
\centering
\begin{tabular}{c c c c c c c c c c c c c c}
\toprule
 &  & \multicolumn{4}{c}{1000} & \multicolumn{4}{c}{2000} & \multicolumn{4}{c}{4000} \\
\cmidrule(lr){3-6}
\cmidrule(lr){7-10}
\cmidrule(lr){11-14}
$q$ & $\lambda$ & \multicolumn{2}{c}{SD} & \multicolumn{2}{c}{RMSE} & \multicolumn{2}{c}{SD} & \multicolumn{2}{c}{RMSE} & \multicolumn{2}{c}{SD} & \multicolumn{2}{c}{RMSE} \\
\cmidrule(lr){3-4}
\cmidrule(lr){5-6}
\cmidrule(lr){7-8}
\cmidrule(lr){9-10}
\cmidrule(lr){11-12}
\cmidrule(lr){13-14}
 &  & IKMW & Eff & IKMW & Eff & IKMW & Eff & IKMW & Eff & IKMW & Eff & IKMW & Eff \\
\midrule
\multirow{3}{*}{1} & .33 & .37 & .35 & .39 & .37 & .27 & .26 & .28 & .27 & .19 & .18 & .20 & .20 \\
 & .50 & .47 & .40 & .48 & .41 & .30 & .27 & .31 & .28 & .22 & .20 & .23 & .21 \\
 & .67 & .80 & .46 & .80 & .47 & .52 & .33 & .52 & .34 & .30 & .23 & .30 & .24 \\
\midrule
\multirow{3}{*}{2} & .33 & .37 & .34 & .39 & .37 & .27 & .25 & .29 & .28 & .19 & .18 & .21 & .20 \\
 & .50 & .47 & .38 & .48 & .39 & .32 & .26 & .33 & .29 & .23 & .20 & .25 & .22 \\
 & .67 & .84 & .44 & .84 & .46 & .55 & .31 & .56 & .34 & .33 & .22 & .33 & .24 \\
\bottomrule
\end{tabular}
\caption{Simulation results.}
\label{tab:sim}
\caption*{\footnotesize Simulation results for the estimator in \cite{imbens2025long} (IKMW) and the efficient estimator proposed in Section \ref{sec:long-term}. The number of replications is 500. The parameter $q$ denotes the transformation applied to the variables in (\ref{eqn:q}) to introduce nonlinearity, and $\lambda$ is the share of experimental data in the sample.}
\end{table}

\section{Conclusion} \label{sec:conclude}

This paper develops a unified framework for analyzing local identification in modern causal inference by explicitly linking structural models to their corresponding observable statistical models. We show that the general nonparametric treatment model under unconfoundedness is locally just identified; consequently, all regular, asymptotically linear estimators are first-order equivalent, and the influence functions are all first-order equivalent to the efficient influence function. In contrast, models that are locally overidentified—such as those involving negative controls or long-term causal inference—allow for efficiency improvements and enable nontrivial specification tests.

These results highlight the central role of local overidentification in modern causal analysis. When developing new causal models, it is essential to carry out an explicit model identification analysis. If the model is locally just identified, effort should focus on higher-order properties of regular estimators as they are all first-order equivalent. When feasible, it is preferable to formulate locally overidentified models, which permit the construction of semiparametric efficient estimators, support specification testing with non-trivial powers, and yield richer empirical content.

For future research, the framework could be extended to accommodate dependent data structures and, for example, to investigate event study approaches to causal inference in time series settings \citep{linton2007quantilogram,han2016cross}.

\appendix
\newpage

\numberwithin{equation}{section} 

\section{Technical proofs}

\begin{proof}[Proof of Theorem \ref{thm:ATE-basic}]
    Take any $P \in \mathbf{M}(\mathcal{Y}^\mathcal{T} \times \mathcal{T}\times \mathcal{X})$, we want to show that $P \in \mathbf{P}_{\operatorname{UC}}$. To show this, we construct a structural distribution $P^* \in \mathbf{P}^*_{\operatorname{UC}}$ such that $P$ is induced by $P^*$, that is, $P = P^* \circ s^{-1}$. Consider $P^*$ defined in the following way.
    
    Let $P^*$ and $P$ have the same marginal distribution for $X$. Then we construct the conditional distribution of $(\{Y(t)\}_{t \in \mathcal{T}},T)$ given $X$. We first examine the case of a binary treatment with $\mathcal{T} = \{0,1\}$ and then extend it to general $\mathcal{T}$. We specify the conditional distribution of $(\{Y(t)\}_{t \in \mathcal{T}},T)|X$ as follows:
    \begin{align*}
        & (Y(1),Y(0)) \perp T|X, \text{ and } Y(1) \perp Y(0) \mid X \text{ under } P^*, \\
        & \mathbb{P}_{P^*}(T=t|X) = \mathbb{P}_{P}(T=t|X), \text{ for } t \in \{0,1\}, \\
        & \mathbb{P}_{P^*}(Y(1) \in B|X) = \mathbb{P}_{P}(Y \in B | T=1,X), \text{ for any measurable } B \subset \mathcal{Y}, \\
        & \mathbb{P}_{P^*}(Y(0) \in B|X) = \mathbb{P}_{P}(Y \in B | T=0,X), \text{ for any measurable } B \subset \mathcal{Y},
    \end{align*}
    where we use $\mathbb{P}_{P^*}$ and $\mathbb{P}_{P}$ to signify that the probability is taken under the distribution $P^*$ and $P$, respectively. It is straightforward to verify that the distribution $P^*$ constructed in this way is a probability measure and satisfies the unconfoundedness condition. Hence, $P^* \in \mathbf{P}^*_{\operatorname{UC}}$.
    We can verify that for any $B_1 \subset \mathcal{Y}$ and $B_2 \subset \mathcal{X}$, we have
    \begin{align*}
        & P^* \circ s^{-1}(B_1 \times \{1\} \times B_2) \\
        =& P^* (\{ (y(1),y(0),1) : y(1) \in B_1\} \times  B_2 ) \\
        =& P^*(B_1 \times \mathcal{Y} \times \{1\} \times B_2) \\
        =& \mathbb{P}_{P^*}(Y(1) \in B_1|X \in B_2) \mathbb{P}_{P^*}(Y(0) \in \mathcal{Y}|X \in B_2) \mathbb{P}_{P^*}(T=1|X \in B_2)\mathbb{P}_{P^*}(X \in B_2) \\
        =& \mathbb{P}_{P}(Y \in B_1 | T=1,X\in B_2) \times 1 \times  \mathbb{P}_{P}(T=1|X) \times \mathbb{P}_{P}(X \in B_2) \\
        =& \mathbb{P}_{P}(Y \in B_1 , T=1,X\in B_2) \\
        =& P(B_1 \times \{1\} \times B_2),
    \end{align*}
    where the product decomposition in the third equality is due to the construction that $(Y_1,Y_0) \perp T|X, \text{ and } Y_1 \perp Y_0|X \text{ under } P^*$.
    Similarly, we can verify that $P^* \circ s^{-1}(B_1 \times \{0\}\times B_2) = P(B_1 \times \{0\}\times B_2)$. This shows that $P = P^* \circ s^{-1}$ as desired, which proves that $\mathbf{P}_{\operatorname{UC}} = \mathbf{M}(\mathcal{Y} \times \{0,1\}\times \mathcal{X})$.

    Now we extend to a general index set $\mathcal T\subset\mathbb R$. Assume $\mathcal Y$ is a standard Borel space and $\mathcal T$ is a Borel subset of $\mathbb R$. For each $x\in\mathcal X$ and finite tuple $\mathbf t=(t_1,\dots,t_k)\in\mathcal T^k$, define $Q_x^{\mathbf t}(B_1\times\cdots\times B_k)=\prod_{j=1}^k \mathbb P_{P}(Y\in B_j\mid T=t_j,X=x)$, $B_j\in\mathcal B(\mathcal Y)$. This family is symmetric and consistent under marginalization, so by the Kolmogorov extension theorem there exists a probability measure $Q_x$ on $(\mathcal Y^{\mathcal T},\mathcal B(\mathcal Y)^{\otimes\mathcal T})$ whose coordinates $\{Y(t)\}_{t\in\mathcal T}$ are (conditionally on $X=x$) mutually independent with marginals $Y(t)\mid X=x\sim\mathbb P_{P}(Y\in\cdot\mid T=t,X=x)$. Define $P^*$ by first drawing $X\sim P_X$, then, given $X=x$, drawing $T\sim\mathbb P_{P}(T\in\cdot\mid X=x)$ and an independent process $\{Y(t)\}_{t\in\mathcal T}\sim Q_x$. Under $P^*$, unconfoundedness holds: $\{Y(t)\}_{t\in\mathcal T}\perp T\mid X$. Moreover, for measurable $B\subset\mathcal Y$, $C\subset\mathcal X$, and $t\in\mathcal T$,
\begin{align*}
    P^*\!\circ s^{-1}(B\times\{t\}\times C)&=\int_{C}\mathbb P_{P}(T=t\mid X=x)\,\mathbb P_{P}(Y\in B\mid T=t,X=x)\,P_X(dx)\\
    & =P(B\times\{t\}\times C).
\end{align*}
A standard $\pi$–$\lambda$ argument extends this equality from rectangles to the full product $\sigma$-algebra, hence $P=P^*\!\circ s^{-1}$. Therefore, the argument for $\mathcal T=\{0,1\}$ carries over verbatim to any Borel $\mathcal T\subset\mathbb R$.

\end{proof}

\begin{proof}[Proof of Lemma \ref{lm:neg-contr-obs-equiv}]
Fix any observable law $P$ of $(Y,D,V,Z,X)$ for which there exists a measurable function $h$ satisfying (\ref{eqn:neg-contr-bridge}). We construct, on an extension of the underlying probability space, a latent variable $U$ and potential variables $\{Y(d,z),V(d,z),D(z):d\in\{0,1\},z\in\mathcal Z\}$ such that the identification assumptions are satisfied and the induced observable probability law is $P$.

Extend the space to carry three iid $\mathrm{Unif}(0,1)$ variables $U_D,U_V,U_Y$ independent of $(Y,D,V,Z,X)$ under $P$. Set $U=Z$. This choice makes the completeness condition hold automatically. Define the potential treatment as $D(z) = \mathbf{1}\{\mathbb{P}_P(D=1|Z=z,X) \geq U_D\}$. Then $D=D(Z)$ and, by construction, the conditional law of $D$ given $(Z,X)$ under the constructed model equals its observed conditional law.

Then we use a conditional–quantile construction so that $(Y(d),V)\mid(U,X)$ has exactly the observed joint law of $(Y,V)\mid(D=d,Z=U,X)$, while keeping $V(d,z)=V$. Formally, let $F_{V\mid U,X}(\cdot\mid u,x)$ be the observed conditional cdf of $V$ given $(Z,X)=(u,x)$ and let $V \;=\; F_{V\mid U,X}^{-1}(U_V\mid U,X),$
with $U_V\sim\mathrm{Unif}(0,1)$ independent of $(U,X)$. For each $d\in\{0,1\}$, define the conditional cdf of $Y$ given $(V,U,X,D=d)$ from the observed data by
$
G_d(y\mid v,u,x)\;=\;\mathbb{P}\big(Y\le y\mid V=v,\;D=d,\;Z=u,\;X=x\big).
$
Let $U_Y\sim\mathrm{Unif}(0,1)$ be independent of $(U_V,U,X)$, and set
$
Y(d)\;=\; G_d^{-1}\big(U_Y\mid V,\;U,\;X\big), Y(d,z)=Y(d).
$
Then, for any $u,x$,
$
\mathbb{P}\big(Y(d)\le y,\; V\le v \mid U=u,X=x\big)
=\int_{(-\infty,w]} G_d(y\mid v',u,x)\, dF_{V\mid U,X}(v'\mid u,x)
=F_{Y,V\mid D=d,Z,U,X}(y,v\mid d,u,u,x),
$
so $(Y(d),V)\mid(U,X)$ matches the observed joint distribution of $(Y,V)\mid(D=d,Z=U,X)$, hence reproduces its correlation (and any higher-order) structure. Because $(Y(d),V)$ are measurable functions of $(U,X,U_Y,U_V)$ only, while $(D,Z)$ depend on $(U,X,U_D)$ only, with $(U_Y,U_V,U_D)$ mutually independent and independent of $(U,X)$, we retain $(Y(d),V)\perp(D,Z)\mid(U,X)$. Moreover, $V(d,z)=V$ by construction, and the realized $(Y,D,V,Z,X)$ have exactly the original joint law, so any observable moments—including $\mathbb{E}[Y-h(V,D,X)\mid Z,D,X]=0$—are preserved.
\end{proof}

\begin{proof}[Proof of Theorem \ref{thm:neg-contr-justid}]
    By Lemma \ref{lm:neg-contr-obs-equiv}, we only need to focus on the moment condition (\ref{eqn:neg-contr-bridge}).
    Let $m(Z,D,X,h) = \mathbb{E}_P[Y-h(V,D,X)|Z,D,X]$ for any $h \in L^2(V,D,X)$. The derivative of $m(Z,D,X,\cdot)$ at $h_P$ along any direction $h \in L^2(V,D,X)$ is 
    \begin{align*}
        \nabla m(Z,D,X,h_P)[h] & = \frac{\partial}{\partial \tau} m(Z,D,X,h_P+\tau h) \big|_{\tau=0} \\
        & = -\mathbb{E}_P[\mathbf{1}\{G=O\}h(V,D,X)|Z,D,X] = - Th.
    \end{align*}
    The range space of the linear map $\nabla m(Z,D,X,h_P)$ is given by
    \begin{align*}
        \mathcal{R} & = \left\{ f \in L^2(Z,D,X) : f=-Th, h \in L^2(V,D,X)  \right\} \\
        & = \left\{ f \in L^2(Z,D,X) : f=Th, h \in L^2(V,D,X)  \right\},
    \end{align*}
    which equals the range space of the conditional expectation operator $T$, $\text{Ran}(T)$. By Theorem 4.1 in \cite{chen2018overidentification}, the model specified by the above moment condition is locally just identified if and only if the closure of the range space, $\Bar{\mathcal{R}} = \text{Ran}(T)$, is equal to the full space $L^2(Z,D,X)$. Since the orthogonal complement of $\text{Ran}(T)$ is equal to, $\text{Ker}(T^*)$ \citep[e.g., Theorem 4.12 in][]{rudin1991functional}, the kernel of its adjoint, local just-identification of the model is equivalent to $\text{Ker}(T^*) = \{0\}$. This completes the proof.
\end{proof}

\begin{proof}[Proof of Theorem \ref{thm:neg-contr-speb}]
        Since the first two unconditional moments independently define $\mu_1$ and $\mu_0$, we can, without loss of generality, derive efficiency results for each parameter separately. We begin with $\mu_1$; the derivation for $\mu_0$ is analogous. Afterward, we combine their efficient influence functions to obtain that for ATE. In the case for $\mu_1$, the model has one unconditional moment function $h_1(V,X) - \mu_1$ and two conditional moments in the function $\rho$ with conditioning variables $(Z,X)$. We first need to go through the orthogonalization step in \cite{ai2012semiparametric}. Let $\rho_1 = h_1(V,X) - \mu_1 - \Gamma_1 \rho$. Define $m_1(\mu_1,h) = \mathbb{E}[\rho_1]$ and $m_2(\mu_1,h;Z,X) = \mathbb{E}[\rho_2|Z,X]$. The derivatives of $m_1$ and $m_2$ in $\mu_1$ and $h$ are respectively given by
    \begin{alignat*}{2}
\frac{d m_1}{d\mu_1} &\;=\; -1, &\qquad
\frac{d m_1}{dh}[r] &\;=\; \mathbb{E}\!\bigl[(e_1 - \Gamma_1 L)r\bigr],\\
\frac{d m_2}{d\mu_1} &\;=\; 0, &\qquad
\frac{d m_2}{dh}[r] &\;=\; -L T r.
\end{alignat*}
    where $r$ is any direction in the nonparametric space of $h$. Let $\Sigma_1 = \mathbb{E}[\rho_1^2]$ denote the variance of $\rho_1$. The Fisher information is obtained by minimizing the following objective function over all $r$,
    \begin{align*}
    	& \mathbb{E} \left[ \left(1 + \mathbb{E} \left[ (e_1 - \Gamma_1 L) r \right] \right)^2 \Sigma_1^{-1}  + (LTr)' \Sigma^{-1} (LTr) \right] \\
	= & \big(1 + \mathbb{E} [ \underbrace{T^*(e_1 - \Gamma_1 L)}_{\equiv \psi} r ] \big)^2 \Sigma_1^{-1}  + \langle r, \underbrace{(LT^* \Sigma_2^{-1} TL)}_{\equiv H} r \rangle ,
    \end{align*}
    where the second follows from the derivation that
    \begin{align*}
        (LTr)' \Sigma^{-1} (LTr) = \langle T L r,  \Sigma_2^{-1} TL r \rangle = \langle L r,  T^*\Sigma_2^{-1} TL r \rangle = \langle r,  L T^*\Sigma_2^{-1} TL r \rangle
    \end{align*}
    because $L$ is a function of $(Z,X)$.
    The optimal $r^*$ satisfies the following first-order condition:
    \begin{align*}
    	0 & = \left(1 + \mathbb{E} \left[ \psi r^* \right] \right) \Sigma_1^{-1} \langle r, \psi\rangle + \langle r,  H r^* \rangle,
    \end{align*}
    for any perturbation $r$. This implies that
    \begin{align*}
    	\left(1 + \mathbb{E} \left[ \psi r^* \right] \right) \Sigma_1^{-1} \psi' + H r^* = 0.
    \end{align*}
    Denote $A =1 + \mathbb{E} \left[ \psi r^* \right]$, and let $H^{-1}$ be the generalized inverse of the operator $H$. Then we have 
    \begin{align*}
        r^* = - A \Sigma_1^{-1} H^{-1} \psi',
    \end{align*}
    and hence $A$ satisfies the following equality
    \begin{align*}
         A & = 1 - A \Sigma_1^{-1} \mathbb{E}[\psi H^{-1} \psi'] = 1 - A \Sigma_1^{-1} \lVert H^{-1/2} \psi \rVert^2,
    \end{align*}
    with $\lVert \cdot \rVert$ being the $L^2$-norm in the Hilbert space.
    Therefore, $A$ is solved to be
    \begin{align*}
        A = \frac{\Sigma_1}{\Sigma_1 + \lVert H^{-1/2} \psi' \rVert^2}.
    \end{align*}
    The optimal direction $r^*$ is
    \begin{align*}
        r^* = - \frac{H^{-1} \psi'}{\Sigma_1 + \lVert H^{-1/2} \psi' \rVert^2} ,
    \end{align*}
    The term $\langle r^*, H r^* \rangle$ in the Fisher information can be calculated using the first-order condition:
    \begin{align*}
        \langle r^*, H r^* \rangle = - A \Sigma_1^{-1} \left(A-1 \right).
    \end{align*}
    Hence, the Fisher information is 
    \begin{align*}
        A^2 \Sigma_1^{-1} - A \Sigma_1^{-1} \left(A-1 \right)
        = & A \Sigma_1^{-1} = \frac{1}{\Sigma_1 + \lVert H^{-1/2} \psi' \rVert^2}.
    \end{align*}
    The efficiency bound is the inverse of the Fisher information: $\Sigma_1 + \lVert H^{-1/2} \psi' \rVert^2$.
    
    The efficient score and efficient influence function can also be derived following the proof of Theorem 2.1 in \cite{ai2012semiparametric}. The efficient score is
    \begin{align*}
        & - \left( \frac{dm_1}{d\mu_1} - \frac{dm_1}{dh}[r^*] \right)' \Sigma_1^{-1} \rho_1 - \left( \frac{dm_2}{d\mu_1} - \frac{dm_2}{dh}[r^*] \right)' \Sigma_2^{-1} \rho_2 \\
        = & \frac{\rho_1 - (LTH^{-1}\psi')'\Sigma^{-1}\rho}{\Sigma_1 + \lVert H^{-1/2} \psi' \rVert^2}.
    \end{align*}
    The efficient influence function is equal to the efficient variance bound multiplied by the efficient score:
    \begin{align*}
        \mathbb{EIF}(\mu_1) & = h_1(V,X) - \mu_1 - \Gamma_1 \rho - (LTH^{-1}\psi')'\Sigma^{-1}\rho.
    \end{align*}
    Similarly, the efficient influence function for $\mu_0$ is 
    \begin{align*}
        \mathbb{EIF}(\mu_0) & = h_0(V,X) - \mu_0 - \Gamma_0 \rho - (LTH^{-1}\psi')'\Sigma^{-1}\rho.
    \end{align*}
    This delivers the expression for the efficient influence function presented in the theorem.

    Lastly, we examine the case in which $T$ and $T^*$ are bijective. In this case, $H$ is an invertible operator, and we have $L T H^{-1} T^* = \Sigma L^{-1}$. Therefore, the efficient influence functions reduce to
\begin{align*}
    \mathbb{EIF}(\mu_1) & = \frac{D\bigl(Y-h_1(V,X)\bigr)}{p(Z,X)} + h_1(V,X) - \mu_1, \\
    \mathbb{EIF}(\mu_0) & = \frac{(1-D)\bigl(Y-h_0(V,X)\bigr)}{1-p(Z,X)} + h_0(V,X) - \mu_0.
\end{align*}
Note that the factor $1/\mathbb{P}(D=d\mid Z,X)$ coincides with the treatment confounding bridge defined in Theorem~2.2 of \cite{cui2024semiparametric} under the bijectivity of $T$; see their discussion following Theorem~2.2.

\end{proof}

\begin{proof}[Proof of Lemma \ref{lm:long-term-obs-equiv}]
We construct the distribution of $(\{Y(d), S(d): d = 0,1\}, U)$ separately for the observational and experimental populations. On the observational population ($G = O$), define $U^O = (S_1, S_2)$. For the short-term potential outcomes, define $S^O(d) = S$, since all components of $S$ are observed under both treatment statuses. For the long-term potential outcome, define
\begin{align*}
    Y^O(d) = F^{-1}_{Y | S, D, G}(U_{[0,1]} | S, d, O),
\end{align*}
where $U_{[0,1]}$ is a uniform random variable on $[0,1]$, independent of all other variables. This construction ensures that $Y^O(D) = Y$ holds in the observational sample. The function $F^{-1}_{Y | S, D, G}$ is the conditional quantile function of $Y$ given $(S,D,G)$.
Because $S^O(d)$ is a deterministic function of $U^O$, and the only randomness in $Y^O(d)$ comes from the independent uniform variable $U_{[0,1]}$, it follows that $(Y^O(d), S^O(d)) \perp D \mid U^O$, verifying the first condition. Furthermore, conditional on $U^O$, $S_1(d)$ is deterministic and therefore independent of $(Y^O(d), S_3^O(d))$, which verifies the third condition. For the same reason, the completeness condition (the fifth condition) also holds, since all randomness is separated and $U^O$ is a sufficient statistic.

On the experimental population ($G = E$), define $(Y^E(d), S^E(d), U^E)$ to be an independent copy of $(Y^O(d), S^O(d), U^O)$; that is, the two triplets have the same distribution but are mutually independent. In particular, this implies that $(Y^E(d), S^E(d), U^E) \perp D$ in the experimental sample, which satisfies the second condition.
To complete the construction, define the full-sample potential outcomes by combining the two populations:
\begin{align*}
    Y(d) & = \mathbf{1}\{G = O\} Y^O(d) + \mathbf{1}\{G = E\} Y^E(d),\\
S(d) & = \mathbf{1}\{G = O\} S^O(d) + \mathbf{1}\{G = E\} S^E(d), \\
U & = \mathbf{1}\{G = O\} U^O + \mathbf{1}\{G = E\} U^E.
\end{align*}
This guarantees that the conditional distribution of $(S(d), U) \mid G$ is invariant to $G$, verifying the fourth condition.

The first five conditions are thus satisfied. Finally, the sixth condition follows from Lemma 1 in \citet{imbens2025long}, which establishes that any function $h$ satisfying the observable bridge condition \eqref{eqn:outcome-bridge-obs} also satisfies the unobservable bridge condition \eqref{eqn:outcome-bridge-unobs}. All six conditions are therefore verified, and since the observable distribution was arbitrary subject only to (\ref{eqn:outcome-bridge-obs}), that moment condition is indeed the model’s sole refutable implication.
\end{proof}

\begin{proof}[Proof of Theorem \ref{thm:long-term-justid}]
    By Lemma \ref{lm:long-term-obs-equiv}, we only need to focus on the moment condition (\ref{eqn:outcome-bridge-obs}), which, under the condition that $\mathbb{P}(G=O|S_2,S_1,D)>0$, is equivalent to the following moment condition: 
    \begin{align*}
        \mathbb{E}[\mathbf{1}\{G=O\}(Y-h(S_3,S_2,D))|S_2,S_1,D]=0.
    \end{align*}
    Let $m(S_2,S_1,D,h) = \mathbb{E}_P[\mathbf{1}\{G=O\}(Y-h(S_3,S_2,D))|S_2,S_1,D]$ for any $h \in L^2(S_3,S_2,D)$. The derivative of $m(S_2,S_1,D,\cdot)$ at $h_P$ along any direction $h \in L^2(S_3,S_2,D)$ is 
    \begin{align*}
        \nabla m(S_2,S_1,D,h_P)[h] & = \frac{\partial}{\partial \tau} m(S_2,S_1,D,h_P+\tau h) \big|_{\tau=0} \\
        & = -\mathbb{E}_P[\mathbf{1}\{G=O\}h(S_3,S_2,D)|S_2,S_1,D] = - Kh.
    \end{align*}
    The range space of the linear map $\nabla m(S_2,S_1,D,h_P)$ is given by
    \begin{align*}
        \mathcal{R} & = \left\{ f \in L^2(S_2,S_1,D) : f=-Kh, h \in L^2(S_3,S_2,D)  \right\} \\
        & = \left\{ f \in L^2(S_2,S_1,D) : f=Kh, h \in L^2(S_3,S_2,D)  \right\},
    \end{align*}
    which equals the range space of the conditional expectation operator $K$, $\text{Ran}(K)$. The remainder of the proof parallels that of Theorem \ref{thm:neg-contr-justid} and relies on Theorem 4.1 of \cite{chen2018overidentification}.

\end{proof}

\begin{proof}[Proof of Theorem \ref{thm:long-term-speb}]
    Since the first two unconditional moments independently define $\mu_1$ and $\mu_0$, we can, without loss of generality, derive efficiency results for each parameter separately. We begin with $\mu_1$; the derivation for $\mu_0$ is analogous. Afterward, we combine their efficient influence functions to obtain that for ALTTE. In the case for $\mu_1$, the model has one unconditional moment function $\rho_1$ and one conditional moment function $\rho_2$:
    \begin{align*}
    	\rho_1(\theta,h;Y,S,D,G) & = \mathbf{1}\{G=E\} D\left( h(S_3,S_2,D) - \mu_1 \right), \\
	\rho_2(\theta,h;Y,S,D,G) & = \mathbf{1}\{G=O\}(Y - h(S_3,S_2,D)),
    \end{align*}
    where the conditioning variables for the conditional moment are $S_2,S_1,$ and $D$. Notice that $\rho_1$ and $\rho_2$ are orthogonal because $ \mathbf{1}\{G=E\}  \mathbf{1}\{G=O\}=0$. So we do not need to go through the orthogonalization step in \cite{ai2012semiparametric}. Define $m_1(\mu_1,h) = \mathbb{E}[\rho_1(\theta,h;Y,S,D,G)]$ and $m_2(\mu_1,h;S_2,S_1,D) = \mathbb{E}[\rho_2(\mu_1,h;Y,S,D,G)|S_2,S_1,D]$. The derivatives of $m_1$ and $m_2$ in $\mu_1$ and $h$ are respectively given by
    \begin{alignat*}{2}
\frac{d m_1}{d\mu_1} &\;=\; -\mathbb{P}(G=E,D=1), &\qquad
\frac{d m_1}{dh}[r] &\;=\; \mathbb{E}\!\bigl[\mathbf{1}\{G=E\}\, D\, r(S_3,S_2,D)\bigr],\\
\frac{d m_2}{d\mu_1} &\;=\; 0, &\qquad
\frac{d m_2}{dh}[r] &\;=\; -K r.
\end{alignat*}
    where $r$ is any direction in the nonparametric space of $h$. The Fisher information is obtained by minimizing the following objective function over all $r$,
    \begin{align*}
    	& \mathbb{E} \left[ \left(\mathbb{P}(G=E,D=1) + \mathbb{E} \left[ \mathbf{1}\{G=E\} D r(S_3,S_2,D) \right] \right)^2 \Sigma_1^{-1}  + (Kr)^2 \Sigma_2^{-1}(S_2,S_1,D) \right] \\
	= & \left(\mathbb{P}(G=E,D=1) + \mathbb{E} \left[ \tilde{D} r(S_3,S_2,D) \right] \right)^2 \Sigma_1^{-1}  + \langle r, (K^* \Sigma_2^{-1}(S_2,S_1,D) K) r \rangle ,
    \end{align*}
    where $\langle \cdot, \cdot \rangle$ denotes the Hilbert space inner product, $\Sigma_1 = \mathbb{E}[\rho_1^2]$, $\Sigma_2 = \mathbb{E}[\rho_2^2|S_2,S_1,D]$, and $\tilde{D} = \pi(S_3,S_2,D) D$. 
    The optimal $r^*$ satisfies the following first-order condition:
    \begin{align*}
    	0 & = \left(\mathbb{P}(G=E,D=1) + \mathbb{E} \left[ \tilde{D} r^*(S_3,S_2,D) \right] \right) \Sigma_1^{-1} \langle \tilde{D}, r(S_3,S_2,D) \rangle \\
        & + \langle (K^* \Sigma_2^{-1}(S_2,S_1,D) K) r^*, r \rangle,
    \end{align*}
    for any perturbation $r$. This implies that
    \begin{align*}
    	\left(\mathbb{P}(G=E,D=1) + \mathbb{E} \left[ \tilde{D} r^*(S_3,S_2,D) \right] \right) \Sigma_1^{-1} \tilde{D} + (K^* \Sigma_2^{-1}(S_2,S_1,D) K) r^* = 0.
    \end{align*}
    Denote $A = \mathbb{P}(G=E,D=1) + \mathbb{E} \left[ \tilde{D} r^*(S_3,S_2,D) \right]$, and let $(K^* \Sigma_2^{-1}(S_2,S_1,D) K)^{-1}$ be the generalized inverse of the operator $K^* \Sigma_2^{-1}(S_2,S_1,D) K$. Then we have 
    \begin{align*}
        r^* = - A \Sigma_1^{-1} (K^* \Sigma_2^{-1}(S_2,S_1,D) K)^{-1} \tilde{D},
    \end{align*}
    and hence $A$ satisfies the following equality
    \begin{align*}
         A & = \mathbb{P}(G=E,D=1) - A \Sigma_1^{-1} \mathbb{E}[(K^* \Sigma_2^{-1}(S_2,S_1,D) K)^{-1} \tilde{D}^2] \\
        & = \mathbb{P}(G=E,D=1) - A \Sigma_1^{-1} \lVert (K^* \Sigma_2^{-1}(S_2,S_1,D) K)^{-1/2} \tilde{D} \rVert^2.
    \end{align*}
    Therefore, $A$ is solved to be
    \begin{align*}
        A = \frac{\mathbb{P}(G=E,D=1) \Sigma_1}{(\Sigma_1 + \lVert (K^* \Sigma_2^{-1}(S_2,S_1,D) K)^{-1/2} \tilde{D} \rVert^2)}.
    \end{align*}
    The optimal direction $r^*$ is
    \begin{align*}
        r^* = - \frac{\mathbb{P}(G=E,D=1)(K^* \Sigma_2^{-1}(S_2,S_1,D) K)^{-1} \tilde{D}}{(\Sigma_1 + \lVert (K^* \Sigma_2^{-1}(S_2,S_1,D) K)^{-1/2} \tilde{D} \rVert^2)} ,
    \end{align*}
    The term $\langle r^*, (K^* \Sigma_2^{-1}(S_2,S_1,D) K) r^* \rangle$ in the Fisher information can be calculated using the first-order condition:
    \begin{align*}
        \langle r^*, (K^* \Sigma_2^{-1}(S_2,S_1,D) K) r^* \rangle = - A \Sigma_1^{-1} \left(A-\mathbb{P}(G=E,D=1) \right).
    \end{align*}
    Hence, the Fisher information is 
    \begin{align*}
        A^2 \Sigma_1^{-1} - A \Sigma_1^{-1} \left(A-\mathbb{P}(G=E,D=1) \right)
        = & A \Sigma_1^{-1}/\mathbb{P}(G=E,D=1) \\
        = & \frac{\mathbb{P}(G=E,D=1)^2}{(\Sigma_1 + \lVert (K^* \Sigma_2^{-1}(S_2,S_1,D) K)^{-1/2} \tilde{D} \rVert^2)}.
    \end{align*}
    The efficiency bound is the inverse of the Fisher information:
    \begin{align*}
        \frac{(\Sigma_1 + \lVert (K^* \Sigma_2^{-1}(S_2,S_1,D) K)^{-1/2} \pi(S_3,S_2,D) D \rVert^2)}{\mathbb{P}(G=E,D=1)^2}.
    \end{align*}
    The efficient score and efficient influence function can also be derived following the proof of Theorem 2.1 in \cite{ai2012semiparametric}. The efficient score is
    \begin{align*}
        & - \left( \frac{dm_1}{d\mu_1} - \frac{dm_1}{dh}[r^*] \right) \Sigma_1^{-1} \rho_1 - \left( \frac{dm_2}{d\mu_1} - \frac{dm_2}{dh}[r^*] \right) \Sigma_2^{-1} \rho_2 \\
        = & \frac{\mathbb{P}(G=E,D=1)\mathbf{1}\{G=E\} D\left( h(S_3,S_2,D) - \mu_1 \right)}{(\Sigma_1 + \lVert (K^* \Sigma_2^{-1}(S_2,S_1,D) K)^{-1/2} \tilde{D} \rVert^2)} \\
        + & \frac{ \mathbb{P}(G=E,D=1)\mathbf{1}\{G=O\}[K(K^* \Sigma_2^{-1}(S_2,S_1,D) K)^{-1} \tilde{D}] \Sigma_2^{-1}(Y - h(S_3,S_2,D))}{(\Sigma_1 + \lVert (K^* \Sigma_2^{-1}(S_2,S_1,D) K)^{-1/2} \tilde{D} \rVert^2)}.
    \end{align*}
    The efficient influence function is equal to the efficient variance bound multiplied by the efficient score:
    \begin{align*}
        \mathbb{EIF}(\mu_1) & = \frac{\mathbf{1}\{G=E\} D\left( h(S_3,S_2,D) - \mu_1 \right)}{\mathbb{P}(G=E,D=1)} \\
        & + \frac{ \mathbf{1}\{G=O\}[K(K^* \Sigma_2^{-1}(S_2,S_1,D) K)^{-1} \pi(S_3,S_2,D)] D \Sigma_2^{-1}(Y - h(S_3,S_2,D))}{\mathbb{P}(G=E,D=1)}.
    \end{align*}
    Here we can take the variable $D$ outside of the operator $K$ because $D$ is measurable with respect to the sigma-algebra generated by $(S_2,S_1,D)$.
    Similarly, we can obtain the efficient influence function for $\mu_0$ as
    \begin{align*}
        \mathbb{EIF}(\mu_0) & = \frac{\mathbf{1}\{G=E\} (1-D)\left( h(S_3,S_2,D) - \mu_1 \right)}{\mathbb{P}(G=E,D=0)} \\
        & + \frac{ \mathbf{1}\{G=O\}[K(K^* \Sigma_2^{-1}(S_2,S_1,D) K)^{-1} \pi(S_3,S_2,D)] (1-D) \Sigma_2^{-1}(Y - h(S_3,S_2,D))}{\mathbb{P}(G=E,D=0)}.
    \end{align*}

    Next, we derive the formula for the efficient variance bound of $\mu_1 - \mu_0$ by computing the second moment of the efficient influence function. Note that the efficient influence function can be decomposed into two orthogonal parts with factors $\mathbf{1}\{G=E\}$ and $\mathbf{1}\{G=O\}$. Below, we repeatedly use the fact that $D f(D) = D f(1)$ and $(1-D)f(D) = (1-D)f(0)$ for any function $f$ of $D$. For the experimental part with $G=E$, note that $D\mu_1 = D\mu(D)$ and $(1-D)\mu_0 = (1-D)\mu(D)$. Therefore, we have
    \begin{align*}
        & \frac{\mathbf{1}\{G=E\} D\left( h(S_3,S_2,D) - \mu_1 \right)}{\mathbb{P}(G=E,D=1)} -  \frac{\mathbf{1}\{G=E\} (1-D)\left( h(S_3,S_2,D) - \mu_1 \right)}{\mathbb{P}(G=E,D=0)} \\
        = & \frac{\mathbf{1}\{G=E\}}{\mathbb{P}(G=E)} \left( \frac{D}{p_E} - \frac{1-D}{1-p_E}\right)\left( h(S_3,S_2,D) - \mu(D) \right) \\
        = & \frac{\mathbf{1}\{G=E\}}{\mathbb{P}(G=E)} \frac{D-p_E}{p_E(1-p_E)}\left( h(S_3,S_2,D) - \mu(D) \right).
    \end{align*}
    The second moment of the above expression gives the first term in the efficiency bound. For the observational part with $G=O$, we have
    \begin{align*}
        \frac{\mathbf{1}\{G=O\}}{\mathbb{P}(G=E,D=1)}D & = \frac{\mathbf{1}\{G=O\}}{\mathbb{P}(G=O,D=1)} \frac{\mathbb{P}(G=O,D=1)}{\mathbb{P}(G=E,D=1)} D \\
        & = \frac{\mathbf{1}\{G=O\}}{\mathbb{P}(G=O, D=1)} \frac{\mathbb{P}(G=O | D=1)}{\mathbb{P}(G=E | D=1)} D \\
        & = \frac{\mathbf{1}\{G=O\}}{\mathbb{P}(G=O, D=1)} \frac{\mathbb{P}(G=O | D)}{\mathbb{P}(G=E | D)} D \\
        & = \frac{\mathbf{1}\{G=O\}}{\mathbb{P}(G=O)} \frac{\mathbb{P}(G=O | D)}{\mathbb{P}(G=E | D)} \frac{D}{p_O},
    \end{align*}
    where in the second equality we divide the numerator and denominator by $\mathbb{P}(D=1)$. Similarly, we have
    \begin{align*}
         \frac{\mathbf{1}\{G=O\}}{\mathbb{P}(G=E,D=0)}(1-D) & = \frac{\mathbf{1}\{G=O\}}{\mathbb{P}(G=O)} \frac{\mathbb{P}(G=O | D)}{\mathbb{P}(G=E | D)} \frac{1-D}{1-p_O}.
    \end{align*}
    The entire observational part of the efficient influence function becomes
    \begin{align*}
        \frac{\mathbf{1}\{G=O\}}{\mathbb{P}(G=O)} \bar{D} KM^{-1} \pi(S_3,S_2,D) \Sigma_2^{-1}(Y - h(S_3,S_2,D)),
    \end{align*}
    where $\bar{D} =  \frac{D-p_O}{p_O(1-p_O)} \frac{\mathbb{P}(G=O | D)}{\mathbb{P}(G=E | D)}$.
    The second moment of the above term (omitting the $\mathbb{P}(G=O)^2$ term in the denominator) is
    \begin{align*}
        & \mathbb{E} \left[ \mathbf{1}\{G=O\} \bar{D}^2 [KM^{-1} \pi(S_3,S_2,D)]^2 \Sigma_2^{-2}(Y - h(S_3,S_2,D))^2 \right] \\
        = & \mathbb{E} \left[ \bar{D}^2 [KM^{-1} \pi(S_3,S_2,D)]^2 \Sigma_2^{-2} \mathbb{E}[\mathbf{1}\{G=O\} (Y - h(S_3,S_2,D))^2| S_2,S_1,D] \right] \\
        = & \mathbb{E} \left[ [KM^{-1} \pi(S_3,S_2,D)\bar{D}]^2 \Sigma_2^{-1} \right] \\
        = & \langle KM^{-1} \pi(S_3,S_2,D)\bar{D}, \Sigma_2^{-1}KM^{-1} \pi(S_3,S_2,D)\bar{D} \rangle \\
        = & \langle M^{-1} \pi(S_3,S_2,D)\bar{D}, K^*\Sigma_2^{-1}KM^{-1} \pi(S_3,S_2,D)\bar{D} \rangle \\
        = & \left\lVert M^{-1/2} \bar{D} \pi(S_3,S_2,D) \right\lVert^2,
    \end{align*}
    where the second line follows from that $\bar{D}$ and $[KM^{-1} \pi(S_3,S_2,D)]$ are measurable with respect to $(S_2,S_1,D)$, the third line follows from the definition of $\Sigma_2$, and the fifth line follows from the definition of $K^*$. This proves the formula for the efficiency bound.
    
    Lastly, we verify that the bound reduces to the one obtained by \citet{imbens2025long} in the case where $K$ is bijective. The efficiency bound given in Theorem 7 of \citet{imbens2025long} is
    \begin{align*}
\sigma^{2}
&=\frac{1+\lambda}{\lambda}\,
  \mathbb{E}\Biggl[
    \Bigl(
      \frac{D-\mathbb{P}(D=1 | G=E)}
           {\mathbb{P}(D=1 | G=E)}
      \bigl(h(S_{3},S_{2},D)- \mu(D)\bigr)
    \Bigr)^{2}
    \,\Bigm|\,G=E
  \Biggr] \\
&+(1+\lambda)\,
  \mathbb{E}\Biggl[
    \Bigl(
      \frac{D-\mathbb{P}(D=1 | G=O)}
           {\mathbb{P}(D=1 | G=O)}
      \,q(S_{2},S_{1},D)\,
      \bigl(Y-h(S_{3},S_{2},D)\bigr)
    \Bigr)^{2}
    \,\Bigm|\,G=O
  \Biggr],
\end{align*}
where $\lambda = \mathbb{P}(G=E)/\mathbb{P}(G=O)$, and $q$ is a function satisfying the following equation:
\begin{align*}
    \mathbb{E}\left[ \mathbf{1}\{G=O\} \left(\frac{\mathbb{P}(G=E|D)}{\mathbb{P}(G=O|D)} q(S_2,S_1,D) + 1 \right) \Big| S_2,S_1,D \right] = 1.
\end{align*}
Solving the above equation, we obtain that
\begin{align*}
    q(S_2,S_1,D) = \frac{\mathbb{P}(G=O|D)}{\mathbb{P}(G=E|D)} \frac{\mathbb{P}(G=E | S_2, S_1, D)}{\mathbb{P}(G=O | S_2, S_1, D)}.
\end{align*}
    The first term in $\sigma^2$ already matches the first term of our efficiency bound, and we only need to focus on the second term derived from the observational part of the data.
    Notice that the operator $K(\cdot) = \mathbb{P}(G=O|S_2,S_1,D) \mathbb{E}[\cdot | S_2,S_1,D]$ and its adjoint $K^*(\cdot) = \mathbb{P}(G=O|S_3,S_2,D) \mathbb{E}[\cdot | S_3,S_2,D]$. When the operator $K$ is bijective, both $\mathbb{E}[\cdot | S_2,S_1,D]$ and $\mathbb{E}[\cdot | S_3,S_2,D]$ equal the identity operator.\footnote{This is because, for any bijective idempotent operator $\tilde{K}$, there exists a well-defined inverse $\tilde{K}^{-1}$. We have $\tilde{K}^2 = \tilde{K}$. Apply $\tilde{K}^{-1}$ to both sides of the equation, we obtain that $\tilde{K}$ is the identity operator.} Additionally, there exists an invertible measurable mapping between $(S_3,S_2,D)$ and $(S_2,S_1,D)$. This implies that $\pi(S_3,S_2,D) = \mathbb{P}(G=E|S_2,S_1,D)$. The operator $M$ simplifies to
    \begin{align*}
        M = K^* \Sigma_2^{-1} K = \mathbb{P}(G=O|S_2,S_1,D)^2 \Sigma_2^{-1} I,
    \end{align*}
    with $I$ being the identity operator.
    Therefore, the second term in our efficiency bound simplifies to
    \begin{align*}
       \frac{1}{\mathbb{P}(G=O)^2} \left\lVert \Sigma_2^{1/2}  \bar{D} \frac{\mathbb{P}(G=E|S_2,S_1,D)}{\mathbb{P}(G=O|S_2,S_1,D)} \right\lVert^2,
    \end{align*}
    which is equal to the second term of $\sigma^2$.

\end{proof}

\bibliographystyle{chicago}
\bibliography{localid.bib}

\end{document}